\documentclass[11pt,a4paper]{article}    

\usepackage[margin=1in]{geometry} 
\usepackage{amsmath,amssymb,amsthm}
\usepackage{color}
\usepackage{relsize}
\usepackage[perpage,symbol*]{footmisc}
\usepackage{authblk}
\usepackage{enumerate}
\usepackage{cases}
\usepackage{cite}

\newtheorem{theorem}{Theorem}[section]
\newtheorem{proposition}[theorem]{Proposition}
\newtheorem{lemma}[theorem]{Lemma}

\newtheorem{remark}{Remark}

\newcommand{\be}{\begin{equation}}
	\newcommand{\ee}{\end{equation}}
\newcommand{\bea}{\begin{eqnarray}}
	\newcommand{\eea}{\end{eqnarray}}

\numberwithin{equation}{section}

\title{Hankel Determinants for a Deformed Laguerre Weight with Multiple Variables and Generalized Painlev\'{e} V Equation }
\author[1]{Xinyu Mu}
\author[1,\thanks{Author to whom any correspondence should be addressed. E-mail: lvshulin1989@163.com}]{Shulin Lyu}
\affil[1]{School of Mathematics and Statistics, Qilu University of Technology (Shandong Academy of Sciences), Jinan	250353, China}

\date{}

\begin{document}
	\maketitle
	
\begin{abstract}

We study the Hankel determinant generated by the moments of the deformed Laguerre weight function
$x^{\alpha}{\rm{e}}^{-x}\prod\limits_{k=1}^{N}(x+t_k)^{\lambda_k}$,
where $x\in \left[0,+\infty \right)$, $\alpha,t_k >0, \lambda_k\in\mathbb{R}$ for $k=1,\cdots ,N$.
By using the ladder operators for the associated monic orthogonal polynomials and three compatibility conditions, we express the recurrence coefficients  in terms of the $2N$ auxiliary quantities
which are introduced in the ladder operators and shown to satisfy a system of difference equations.
Combining these results with the differential identities obtained from the differentiation of the orthogonality relations,
we deduce the Riccati equations for the  auxiliary quantities.
From them we establish a system of second order PDEs  which are reduced to a Painlev\'{e} V  equation for $N=1$.
Moreover, we derive the second order PDE satisfied by the logarithmic derivative of the Hankel determinant, from which the limiting PDE is obtained under a double scaling.
When $N=1$, these two PDEs are reduced to the $\sigma$-form of a Painlev\'{e} V and III equation respectively.
When the dimension of the Hermitian matrices from the corresponding deformed  Laguerre unitary ensemble is large and $\lambda_k\ge0$,
based on Dyson's Coulomb fluid theory,  we deduce the equilibrium density for the eigenvalues.

\end{abstract}
\noindent
	
\section{Introduction}	

The classical Laguerre unitary ensemble (LUE for short) is a group of $n\times n$ Hermitian random matrices whose eigenvalues have the following joint probability density function
$$p(x_1,\cdots ,x_n)=\frac{1}{Z_n} \prod_{1\le i< j\le n}(x_i-x_j)^2\prod_{k=1}^{n} w(x_k) ,$$
where $w(x)=x^{\alpha}{\rm{e}}^{-x}$, $\alpha>-1$, $ x_k\in \left[ 0,+\infty \right)$.
The normalization constant $Z_n$, also known as the partition function, is given by \cite[p. 321]{Mehta}
\begin{align*}
	Z_n&=\int_{\left[0,+\infty \right)^n }\prod_{1\le i< j\le n}(x_i-x_j)^2\prod_{k=1}^{n} x_k^{\alpha}{\rm{e}}^{-x_k} dx_1\cdots dx_n \\
	&=\prod_{j=1}^{n} j!\cdot \Gamma(\alpha+j).
\end{align*}

In this paper, we consider the Hankel determinant generated by the moments of a deformation of the classical Laguerre weight function with $N$-variables, namely
\begin{align}
	D_n(\vec{t} \,)={\rm det}\left ( \int_{0}^{+\infty }x^{i+j}w(x;\vec{t}\,)dx  \right ) _{i,j=0}^{n-1} ,\label{Dn}
\end{align}
where $\vec{t}=(t_1,\cdots ,t_N)$ and the weight function reads
\begin{align}
	 w(x;\vec{t}\,)=x^{\alpha}{\rm{e}}^{-x}\prod_{k=1}^{N}(x+t_k)^{\lambda_k},
	\qquad x\in \left[0,+\infty \right),
	\label{w(x)}
\end{align}
with $\alpha,t_k >0, \lambda_k\in\mathbb{R}$ for $k=1,\cdots ,N$.

For $N=1$ and $N=2$, the corresponding Hankel determinants turned out to be connected with the outage capacity and the error probability of multiple-input multiple-output (MIMO for short) wireless communication system.
When $N=1$ in \eqref{w(x)}, the Hankel determinant $D_n(t_1)$ is related to the moment generating function of the outage capacity of a single-user MIMO system \cite{ChenMckay}.
Through the ladder operator approach,  the logarithmic derivative of $D_n(t_1)$ was found to satisfy the $\sigma$-form of a Painlev\'{e} V equation.
Under double scaling, this Hankel determinant was shown in \cite{H.Chen} to have an integral representation in terms of solutions of a Painlev\'{e} III equation.
When $N=2$ in \eqref{w(x)}, the Hankel determinant $D_n(t_1,t_2)$ was intimately related to the moment generating function of the instantaneous received signal-to-noise ratio in MIMO system.
Via the ladder operator approach, an exact representation of the Hankel determinant in terms of a second order PDE was established in \cite{ChenHaq}, which was reduced to a Painlev\'{e} V equation under certain limits.
By adopting Dyson's Coulomb fluid method, a closed-form approximation for the Hankel determinant was derived.

In recent years, MIMO has been applied to fifth-generation (5G) and beyond 5G wireless communication, for example:
terahertz (THz) communication \cite{Kundu1},  visible light communication (VLC) \cite{Narmanlioglu} and radar-communication \cite{Liu}.
The random matrix theory has been widely utilized to study MIMO wireless communication systems \cite{TulinoVerdu,Lopez-Martinez}.
A closed-form approximation of the cumulative distribution function  of the largest eigenvalue of a random Gram matrix $H^{*}H$ associated with the MIMO channel matrix $H$ was derived in \cite{QiQian} by using the Coulomb gas analogy and a Taylor series expansion.
Building on this work, a closed-form approximation for the outage probability of MIMO systems was obtained.
In \cite{Telatar}, the capacities and error exponents of a single user Gaussian channel with multiple transmitting and(or) receiving antennas were studied. It was concluded that
the use of multiple antennas can significantly enhance the achievable rates of fading channels under independence assumptions for the fades and noises at different receiving antennas.

It is well known that the Hankel determinant \eqref{Dn} admits the following expression \cite{Ismail,szego}
\begin{align}
	D_n(\vec{t} \,)=\prod_{k=0}^{n-1} h_k(\vec{t} \,).\label{D_n}
\end{align}
Consider the sequence of monic polynomials $\{P_n(x;\vec{t}\,)\}$ orthogonal with respect to the weight function \eqref{w(x)}, i.e.
\begin{align}
	\int_{0}^{+\infty }P_n(x;\vec{t}\,)P_m(x;\vec{t}\,)w(x;\vec{t}\,)dx=h_n\delta _{mn}, \qquad m,n=0,1,2,\cdots, \label{or}
\end{align}
where $\delta_{mn}=1$ for $m=n$ and $0$ otherwise, and  $P_n(x;\vec{t}\,)$ has the following form
\begin{align}
	P_n(x;\vec{t}\,):=x^n+p(n,\vec{t}\,)x^{n-1}+\cdots +P_n(0;\vec{t}\,), \qquad n\ge1,     \label{P_n(x)}
\end{align}	
and $P_0(x;\vec{t}\,):=1$.
From the above definition, we get the following three-term recurrence relation \cite[p. 22]{Ismail}
\begin{align}
	 xP_n(x;\vec{t}\,)=P_{n+1}(x;\vec{t}\,)+\alpha_n(\vec{t}\,)P_n(x;\vec{t}\,)+\beta_n(\vec{t}\,)P_{n-1}(x;\vec{t}\,), \qquad n=0,1,2\cdots, \label{three-term}
\end{align}	
with the initial conditions $P_0(x;\vec{t}\,):=1$ and $\beta _0P_{-1}(x;\vec{t}\,):=0$. Substituting \eqref{P_n(x)} into \eqref{three-term}, and comparing the coefficient of $x^n$ yields
\begin{align}
	\alpha _n(\vec{t}\,)&=p(n,\vec{t}\,)-p(n+1,\vec{t}\,), \qquad n\ge 0, \label{alpha}
\end{align}
where $p(0,\vec{t}\,):=0$. A telescopic sum of \eqref{alpha} gives us
\begin{align}
	\sum_{k=0}^{n-1}\alpha_k(\vec{t}\,)=-p(n,\vec{t}\,). \label{sum_alpha}
\end{align}
Multiplying $xP_n(x;\vec{t}\,)$ by $P_{n-1}(x;\vec{t}\,)w(x;\vec{t}\,)$ and integrating from $0$ to $+\infty$ with respect to $x$, and combining the orthogonality relation \eqref{or}  with the recurrence relation \eqref{three-term} leads us
\begin{align}
\int_{0}^{+\infty}xP_n(x;\vec{t}\,) P_{n-1}(x;\vec{t}\,)w(x;\vec{t}\,)dx
=\beta_n(\vec{t}\,)h_{n-1}(\vec{t}\,). \label{xP}
\end{align}
Exchanging the positions of $P_n(x;\vec{t}\,)$ and $P_{n-1}(x;\vec{t}\,)$ in the above equation, in view of \eqref{or} and \eqref{three-term}, gives that \eqref{xP} is equal to $h_{n}(\vec{t}\,)$. Hence, we obtain
\begin{align}
 	\beta _n(\vec{t}\,)=\frac{h_n(\vec{t}\,)}{h_{n-1}(\vec{t}\,)}, \qquad n\ge 1. \label{beta}
\end{align}
From the recurrence relation \eqref{three-term}, there follows the
Christoffel-Darboux formula
\begin{align} \sum_{k=0}^{n-1}\frac{P_k(x;\vec{t}\,)P_k(y;\vec{t}\,)}{h_k(\vec{t}\,)}= \frac{P_n(x;\vec{t}\,)P_{n-1}(y;\vec{t}\,)-P_n(y;\vec{t}\,)P_{n-1}(x;\vec{t}\,)}{h_{n-1}(\vec{t}\,)(x-y)}. \label{Darboux}
\end{align}

Using the above properties of the orthogonal polynomials, one can derive a pair of ladder operators satisfied by $P_n(x;\vec{t}\,)$  and three compatibility conditions (numbered (\ref{$S_1$}), (\ref{$S_2$}) and (\ref{$S_{2}'$})) for $A_n(z)$ and $B_n(z)$ which are introduced in the ladder operators. See \cite{vanassche,Chihara,Ismail,szego} for more details of the derivation.
For our problem, we calculate $A_n(z)$ and $B_n(z)$ where $2N$ auxiliary quantities $\{R_{n,k}, r_{n,k}, k=1,\cdots,N\}$ are introduced.
Substituting the resulting expressions into (\ref{$S_1$}), (\ref{$S_2$}) and (\ref{$S_{2}'$}) yields representations for the recurrence coefficients  and  the coefficient of $x^{n-1}$ in $P_n(x;\vec{t}\,)$ in terms of $\{R_{n,k}, r_{n,k}\}$.
Combining these results with the differential relations obtained by taking the derivative of the orthogonality relations,
we finally get a system of $N$ the second order partial differential equations (PDEs for short) for $\{R_{n,k}, k=1,\cdots,N\}$ and an $N$-variable generalization of the $\sigma$-form of a Painlev\'{e} V equation for the logarithmic derivative of the Hankel determinant $D_n(\vec{t}\,)$.
This derivation strategy is known as the ladder operator approach, which has been widely used in the study of Hermitian random matrices
involving one variable \cite{ChenFeigin,MinChenNPB1,MinChenNPB2,MinChenStud,MinChenJMP,ZhuChen} or two variables \cite{LyuGriffinChen,ChenHaq}.

The ladder operator approach was also extensively applied  to study unitary ensembles with multiple variables in the weight function. In \cite{ChenLyu}, the Hankel determinant generated by the Gaussian weight multiplied by a factor with $m$ jump discontinuities  was considered  and its logarithmic derivative was shown to satisfy a second order PDE which was reduced to the $\sigma$-form of a Painlev\'{e} IV equation when $m=1$.
For the Laguerre weight function with $m$ jump discontinuities, the corresponding limiting PDE for the logarithmic derivative of the double-scaled Hankel determinant can be viewed as
an $m$-variable generalization of the $\sigma$-form of a Painlev\'{e} III equation \cite{LyuChenXu}.
The Hankel determinant for the Gaussian weight perturbed by $m$ Fisher-Hartwig singularities was investigated
in \cite{MuLyu} and a second order PDE satisfied by the logarithmic derivative of the Hankel determinant can be reduced to the $\sigma$-form of a Painlev\'{e} IV equation when $N=1$.

Random matrix ensembles with multiple variables can also be studied by using the Riemann-Hilbert method (RH method for short). This is achieved by applying the Deift-Zhou steepest descent analysis to the Riemann-Hilbert problem satisfied by the corresponding orthogonal polynomials.
Deift and Its \cite{DeiftIts} adopted this approach to derive the asymptotics  for $n$-dimensional Toeplitz determinants whose symbols functions have Fisher-Hartwig singularities on the unit circle, and further gave asymptotics of Hankel determinants on a finite interval and also the determinants of Toeplitz+Hankel type.
The asymptotic behavior of large Hankel determinants whose weight functions has a one-cut regular potential and Fisher-Hartwig singularities was studied, which were considered in \cite{Charlier2019} for Gaussian-type ensembles and in \cite{Charlier2021} for Laguerre-type or Jacobi-type ensembles.
In \cite{DaiXuZhang}, the asymptotics of the partition function of the Gaussian unitary ensemble with pole singularities near the soft edge was described by a coupled Painlev\'{e} XXXIV system.
For the application of RH method to random matrix problems
involving one or two perturbation variables in the weight function, see \cite{CharlierDeano,XuZhao,WuXu}.

Recently, the partition function of a modified Gaussian unitary ensemble \cite{DingMin} and the smallest eigenvalue distribution of Freud unitary ensemble \cite{MinWang} were studied. They are intimately related to the determinant of a Hankel matrix, and the main purpose of the researchers is to deduce asymptotics for the logarithmic derivative of the Hankel determinant. This is achieved by conducting the finite-dimensional analysis through the ladder operator approach and the large-dimensional analysis via Dyson's Coulomb fluid method. It should be pointed out that recurrence coefficients play an crucial role in the derivation and no ordinary or partial differential equation could be established, which is different from the former cases.
The outline of the derivation is as follows. By using ({$S_1$}), ({$S_2$}) and ({$S_{2}'$}), the difference equations are obtained for the recurrence coefficients. And with the help of the differentiation of the orthogonality relation, the differential-difference equations are derived for the recurrence coefficients.
In addition, the coefficient of $x^{n-1}$ in $P_n(x;\vec{t}\,)$, which is closely related to the logarithmic derivative of the Hankel determinant, can be expressed in terms of the recurrence coefficients. Via Dyson's Coulomb fluid approach, the large $n$ asymptotics of the recurrence coefficients are derived from the difference equations that they satisfy. And finally the asymptotics of the logarithmic derivative of the Hankel determinant is deduced.

The outline of this paper is as follows. In Section $2$, we make use of the three compatibility conditions (\ref{$S_1$}), (\ref{$S_2$}) and (\ref{$S_{2}'$}) to  express the recurrence coefficients in terms of $2N$ auxiliary quantities $\{R_{n,k}, r_{n,k}, k=1,\cdots,N\}$ which satisfy a system of difference equations that can be iterated in $n$.
In Section $3$, by establishing differential relations, we obtain Toda equations satisfied by the recurrence coefficients and Riccati equations for $\{R_{n,k}, r_{n,k}\}$, and further deduce for $\{R_{n,k}\}$ second order PDEs which is reduced to a Painlev\'{e} V equation when $N=1$.
Section $4$ is devoted to the derivation of the second order PDE satisfied by logarithmic derivative of  $D_n(\vec{t}\,)$, which can be viewed as an $N$-variable generalization of the $\sigma$-form of a Painlev\'{e} V equation. In Section $5$, by sending $n\to \infty$ and $t_k\to 0^+$ such that $s_k=4nt_k$ is fixed, the scaled $\{R_{n,k}\}$ is shown to satisfy a generalized Painlev\'{e} V  equation. We also derive the second order PDE satisfied by the scaled $\sigma_n$, which is the $\sigma$-form of a Painlev\'{e} III equation when $N=1$.
Moreover, when the dimension of the Hermitian matrices from the corresponding deformed  Laguerre unitary ensemble is large and $\lambda_k\ge0$, the equilibrium density is calculated by using Dyson's Coulomb fluid theory.

\section{Ladder Operators and Difference Equations}

In this section, we will apply the ladder operators approach to the Hankel determinant given by \eqref{Dn}.
Before presenting our derivation, we first give the ladder operators satisfied by the associated monic orthogonal polynomials $P_n(z;\vec{t}\,)$ and three compatibility conditions \eqref{$S_1$}, \eqref{$S_2$} and \eqref{$S_{2}'$}. We won't show the dependency of  $\vec{t}\,$ for presentational simplicity unless absolutely required.

\begin{theorem} \label{ladder operators}
The monic orthogonal polynomials $P_n(z;\vec{t}\,)$ satisfy the following lowering and raising operators
\begin{align}
	\biggl (\frac{d}{dz}+B_n(z) \biggr)P_n(z)=\beta_nA_n(z)P_{n-1}(z),\label{2.1}
\end{align}
\begin{align}
	\biggl (\frac{d}{dz}-B_n(z)-{\rm v}'(z) \biggr )P_{n-1}(z)=-A_{n-1}(z)P_n(z),\label{2.2}
\end{align}
	where $A_n(z)$ and $B_n(z)$ are defined by
\begin{align}
	A_n(z):=\frac{1}{h_n} \int_{0}^{+\infty } \frac{{\rm v}'(z)-{\rm v}'(y)}{z-y} P_{n}^{2}(y)w(y)dy,\label{2.3}
\end{align}
\begin{align}
	B_n(z):=\frac{1}{h_{n-1}} \int_{0}^{+\infty } \frac{{\rm v}'(z)-{\rm v}'(y)}{z-y} P_{n}(y)P_{n-1}(y)w(y)dy.\label{2.4}
\end{align}
	Here ${\rm v}(x)={\rm v}(x;\vec{t}\,):=-\ln{w(x;\vec{t}\,)}$.
\end{theorem}

\begin{theorem} \label{compatibility conditions}
	The functions $A_n(z)$ and $B_n(z)$ satisfy the following equations
\begin{equation}
	B_{n+1}(z)+B_n(z)=(z-\alpha _n)A_n(z)-{\rm v}'(z),\tag{$S_1$}\label{$S_1$}
\end{equation}
\begin{equation}
	1+(z-\alpha _n)(B_{n+1}(z)-B_n(z))=\beta _{n+1}A_{n+1}(z)- \beta _{n}A_{n-1}(z),\tag{$S_2$}\label{$S_2$}
\end{equation}
\begin{equation}
	B_{n}^{2}(z)+{\rm v}'(z)B_n(z)+\sum_{j=0}^{n-1}A_j(z)=\beta _nA_n(z)A_{n-1}(z).\tag{$S_{2}'$}\label{$S_{2}'$}
\end{equation}
\end{theorem}	

\begin{remark}
This is the outline of the derivation of Theorem \ref{ladder operators} and Theorem \ref{compatibility conditions}: By using the orthogonality relation \eqref{or} and the Christoffel-Darboux formula \eqref{Darboux}, one derives the lowering operator \eqref{2.1} where $A_n(z)$ and $B_n(z)$ are defined. According to the recurrence relation \eqref{three-term}, the orthogonality relation \eqref{or} and the definitions of $A_n$ and $B_n$, we get \eqref{$S_1$}. Using it and the lowering operator, in view of the recurrence relation \eqref{three-term}, we obtain the raising operator \eqref{2.2}.
Combining it with the lowering operator and the
recurrence relation \eqref{three-term}, we arrive at \eqref{$S_2$}. The combination of \eqref{$S_1$} and \eqref{$S_2$} produces a sum rule \eqref{$S_{2}'$}.
For the detailed derivation process, see \cite{ChenIsmail,Magnus,ChenLyu,MuLyu}.

Another derivation method of the ladder operators is based on the Riemann-Hilbert problem for orthogonal polynomials. According to the  orthogonality relation, with the aid of Cauchy transform, we come to
Theorem \ref{ladder operators}.
The ladder operators can be combined with the recurrence relation \eqref{three-term} to yield \eqref{$S_1$}.
The derivation processes of \eqref{$S_2$} and \eqref{$S_{2}'$} are the same as those in the previous paragraph.
\end{remark}

For our problem, we have
\begin{align}
	{\rm v}(x)=-\ln{w(x;\vec{t}\,)}=x-\alpha\ln{x}-\sum_{k=1}^{N}\lambda_k \ln{(x+t_k)},\label{v(x)}
\end{align}
where ${\rm v}(x)$ is twice continuously differentiable convex functions on $\left[0,+\infty \right)$ and $w(x)$  has moments of all orders.

Substituting \eqref{v(x)} into \eqref{2.3} and \eqref{2.4}, we get the following expressions for $A_n(z)$ and $B_n(z)$.	
\begin{lemma}
	$A_n(z)$ and $B_n(z)$ are given by
\begin{align}
	A_n(z)&= \frac{1-\sum\limits _{k=1}^{N}R_{n,k}}{z}+\sum\limits _{k=1}^{N}\frac{R_{n,k}}{z+t_k},\label{3.2}\\
	B_n(z)&=-\frac{n+\sum\limits _{k=1}^{N}r_{n,k}}{z}+\sum\limits _{k=1}^{N}\frac{r_{n,k}}{z+t_k},\label{3.3}
\end{align}
where the auxiliary quantities $\{R_{n,k}, r_{n,k}, k=1,\cdots,N\}$ are defined by
\begin{align}
R_{n,k}(\vec{t}\, )&:=\frac{\lambda_k }{h_n}\int_{0}^{+\infty } \frac{P_{n}^{2}(y) }{y+t_k}w(y)dy,\label{R_n,k}\\
r_{n,k}(\vec{t}\, )&:=\frac{\lambda_k }{h_{n-1}}\int_{0}^{+\infty } \frac{P_{n}(y)P_{n-1}(y) }{y+t_k}w(y)dy.\label{r_n,k}	
\end{align}
		
\end{lemma}
	
\begin{proof}
According to \eqref{v(x)}, it is easy to see that
\begin{align}
	\frac{{\rm v}'(z)-{\rm v}'(y)}{z-y} =\frac{\alpha }{zy}+\sum_{k=1}^{N}\frac{\lambda_k}{(z+t_k)(y+t_k)}. \label{v'}
\end{align}	
Inserting it into \eqref{2.3} and \eqref{2.4}, we get
\begin{align}
A_n(z)=\frac{1}{z}\cdot \frac{\alpha }{h_n} \int_{0}^{+\infty }\frac{P_n^2(y)}{y}w(y)dy+\sum_{k=1}^{N}\frac{1}{z+t_k}\cdot \frac{\lambda_k}{h_n}\int_{0}^{+\infty }\frac{P_n^2(y)w(y)}{y+t_k}dy,\label{A_n}
\end{align}
\begin{align}
B_n(z)=\frac{1}{z}\cdot\frac{\alpha }{h_{n-1}} \int_{0}^{+\infty }\frac{P_n(y)P_{n-1}(y)}{y}w(y)dy+\sum_{k=1}^{N}\frac{1}{z+t_k}\cdot\frac{\lambda_k}{h_{n-1}}\int_{0}^{+\infty}\frac{P_n(y)P_{n-1}(y)}{y+t_k}w(y)dy.
\label{B_n}
\end{align}
According to integration by parts, and in view of $w(0)=w(+\infty)=0$ , the first term on the right hand side of \eqref{A_n} becomes
\begin{align}
\frac{\alpha }{h_n}\int_{0}^{+\infty  }\frac{P_n^2(y)w(y)}{y}dy=&\frac{1}{h_n}\int_{0}^{+\infty }P_n^2(y){\rm e}^{-y}\prod_{k=1}^{N} (y+t_k) ^{\lambda_k}dy^\alpha \nonumber\\
=&-\frac{1}{h_n} \left[ \int_{0}^{+\infty }2P_n(y)P_n'(y)w(y)dy-\int_{0}^{+\infty }P_n^2(y)y^{\alpha}{\rm e}^{-y}\prod_{k=1}^{N} (y+t_k) ^{\lambda_k}dy  \right.  \nonumber\\
&\left.+\int_{0}^{+\infty}P_n^2(y)y^{\alpha}{\rm e}^{-y} \cdot \frac{\partial }{\partial y} \prod_{k=1}^{N} (y+t_k) ^{\lambda_k} dy	\right]. \label{A_n1.1}
\end{align}
Now we calculate the three integrals in the square bracket one by one. Since the degree of $P_n'(y)$ is $n-1$, in view of orthogonality relation \eqref{or}, we find that the first integral is zero and the second one is $h_n$. Consequently, \eqref{A_n1.1} now reads
\begin{align}
\frac{\alpha }{h_n}\int_{0}^{+\infty  }\frac{P_n^2(y)w(y)}{y}dy=1-\frac{1}{h_n}\int_{0}^{+\infty}P_n^2(y)y^{\alpha}{\rm e}^{-y} \cdot \frac{\partial }{\partial y} \prod_{k=1}^{N} (y+t_k) ^{\lambda_k} dy. \label{A_n1.3}
\end{align}
Taking the derivative of $\prod\limits_{k=1}^{N} (y+t_k) ^{\lambda_k}$ with respect to $y$, we get
\begin{align*}
\frac{\partial }{\partial y} \prod_{k=1}^{N} (y+t_k) ^{\lambda_k}=\sum_{k=1}^{N}\frac{\lambda_k }{y+t_k}\prod_{j=1}^{N} (y+t_j) ^{\lambda_j}.	
\end{align*}
Substituting it back into \eqref{A_n1.3} gives us
\begin{align}
\frac{\alpha }{h_n}\int_{0}^{+\infty  }\frac{P_n^2(y)w(y)}{y}dy=1-\sum_{k=1}^{N}\frac{\lambda_k }{h_n}\int_{0}^{+\infty } \frac{P_{n}^{2}(y) }{y+t_k}w(y)dy. \label{A_n1.2}
\end{align}
Via a similar argument, we find
\begin{align}
\frac{\alpha }{h_{n-1}}\int_{0}^{+\infty  }\frac{P_n(y)P_{n-1}(y)}{y}w(y)dy
=-n-\sum_{k=1}^{N}\frac{\lambda_k }{h_{n-1}}\int_{0}^{+\infty } \frac{P_{n}(y)P_{n-1}(y) }{y+t_k}w(y)dy.\label{B_n1.1}
\end{align}
Putting \eqref{A_n1.2} and \eqref{B_n1.1}  back into equations \eqref{A_n} and \eqref{B_n} respectively, we arrive at the desired expressions \eqref{3.2} and \eqref{3.3}.
\end{proof}
	
Inserting \eqref{3.2} and \eqref{3.3} into (\ref{$S_1$}), we get the following two equations by comparing the coefficients of $z^{-1}$ and $(z+t_k)^{-1}$:
\begin{align}
	&2n+1+\sum_{k=1}^{N}(r_{n+1,k}+r_{n,k})=\alpha _n(1-\sum_{k=1}^{N}R_{n,k})-\alpha,\label{S1.1}\\
	&r_{n+1,k}+r_{n,k}=\lambda_k-(t_k+\alpha _n)R_{n,k},\qquad k=1,\cdots,N. \label{S1.2}
\end{align}
	
Similarly, plugging \eqref{3.2} and \eqref{3.3} into (\ref{$S_2$}), and comparing its both sides the coefficients of  $z^{-1}$ and $(z+t_k)^{-1}$, we obtain
\begin{align}
	\alpha _n+\alpha _n\sum_{k=1}^{N}(r_{n+1,k}-r_{n,k})&=\beta _{n+1}-\beta _n-\sum_{k=1}^{N}(\beta _{n+1}R_{n+1,k}-\beta _{n}R_{n-1,k}),\label{S2.1}\\
	-(t_k+\alpha _n)(r_{n+1,k}-r_{n,k})&=\beta _{n+1}R_{n+1,k}-\beta _{n}R_{n-1,k}, \qquad k=1,\cdots,N.\label{S2.2}
\end{align}

Substituting \eqref{3.2} and \eqref{3.3} into (\ref{$S_{2}'$}) and using the fact that
\begin{align}
	 \frac{1}{z(z+t_k)}=\frac{1}{t_k}\biggl(\frac{1}{z}-\frac{1}{z+t_k}\biggr), \label{z(z+t)}
\end{align}
we get the left side of the equation as below:
\begin{equation}
	\begin{aligned}
l.h.s. =&\frac{1}{z^2}\biggl(n+\sum\limits_{k=1}^{N}r_{n,k}\biggr) \biggl(n+\alpha+\sum\limits_{k=1}^{N}r_{n,k} \biggr)+\biggl(\sum\limits_{k=1}^{N}\frac{r_{n,k}}{z+t_k} \biggr)^2-\sum\limits_{k=1}^{N}\frac{\lambda_k}{z+t_k}\sum\limits_{j=1}^{N}\frac{r_{n,j}}{z+t_j}
\\
&+\sum\limits_{k=1}^{N}\frac{1}{z}\biggl[\frac{1}{t_k}\biggl(\biggl(n+\sum\limits_{j=1}^{N}r_{n,j} \biggr)(\lambda_k-2r_{n,k})-\alpha r_{n,k}\biggr)-r_{n,k}-\sum\limits_{j=0}^{n-1}R_{j,k}\biggr]\\
&+\sum\limits_{k=1}^{N}\frac{1}{z+t_k}\biggl[\frac{1}{t_k}\biggl(\biggl(n+\sum\limits_{j=1}^{N}r_{n,j} \biggr)(2r_{n,k}-\lambda_k)+\alpha r_{n,k}\biggr)+r_{n,k}+\sum\limits_{j=0}^{n-1}R_{j,k}\biggr].\label{lhs}
	\end{aligned}
\end{equation}
Now we analyze the second and third terms on the right side of the above equation. By the polynomial expansion, we find
\begin{align}
	\biggl(\sum\limits_{k=1}^{N}\frac{r_{n,k}}{z+t_k} \biggr)^2
	&=\sum\limits_{k=1}^{N}\frac{r_{n,k}^2}{(z+t_k)^2}
	+\sum_{1\le k<j\le N}\frac{2r_{n,k}r_{n,j}}{(z+t_k)(z+t_j)}\nonumber\\
	 &=\sum\limits_{k=1}^{N}\frac{r_{n,k}^2}{(z+t_k)^2}+\sum_{1\le k<j\le N}\frac{1}{z+t_k}\cdot\frac{2r_{n,k}r_{n,j}}{t_j-t_k}\nonumber\\
	&\qquad\qquad\qquad\quad +\sum_{1\le k<j\le N}\frac{1}{z+t_j}\cdot\frac{2r_{n,k}r_{n,j}}{t_k-t_j}\label{thirdterm}\\
	&=\sum\limits_{k=1}^{N}\frac{r_{n,k}^2}{(z+t_k)^2}+
	\sum_{k\ne j} \frac{1}{z+t_k}\cdot\frac{2r_{n,k}r_{n,j}}{t_j-t_k},\label{lhs1}
\end{align}
where the second equation holds because of the following fact
\begin{align*}
	 \frac{1}{(z+t_k)(z+t_j)}=\frac{1}{t_j-t_k}\biggl(\frac{1}{z+t_k}-\frac{1}{z+t_j}\biggr),
\end{align*}
and the third identity is derived by exchanging $k$ and $j$ in the third term of \eqref{thirdterm}. Via an argument similar to the one used to derive \eqref{lhs1}, we get
\begin{align*}
	 \sum\limits_{k=1}^{N}\frac{\lambda_k}{z+t_k}\sum\limits_{j=1}^{N}\frac{r_{n,j}}{z+t_j}=\sum\limits_{k=1}^{N}\frac{\lambda_kr_{n,k}}{(z+t_k)^2}+	 \sum_{k\ne j} \frac{1}{z+t_k}\cdot\frac{2\lambda_kr_{n,j}}{t_j-t_k}.
\end{align*}
Substituting it and \eqref{lhs1} back into \eqref{lhs} yields
\begin{equation}
\begin{aligned}
l.h.s.=&\frac{1}{z^2}\biggl(n+\sum\limits_{k=1}^{N}r_{n,k}\biggr) \biggl(n+\alpha+\sum\limits_{k=1}^{N}r_{n,k} \biggr)+\sum\limits_{k=1}^{N}\frac{r_{n,k}^2-\lambda_kr_{n,k}}{(z+t_k)^2}
\\
&+\sum\limits_{k=1}^{N}\frac{1}{z}\biggl[\frac{1}{t_k}\biggl(\biggl(n+\sum\limits_{j=1}^{N}r_{n,j} \biggr)(\lambda_k-2r_{n,k})-\alpha r_{n,k}\biggr)-r_{n,k}-\sum\limits_{j=0}^{n-1}R_{j,k}\biggr]    \\
&+\sum\limits_{k=1}^{N}\frac{1}{z+t_k}\biggl[\frac{1}{t_k}\biggl(\biggl(n+\sum\limits_{j=1}^{N}r_{n,j} \biggr)(2r_{n,k}-\lambda_k)+\alpha r_{n,k}\biggr)+r_{n,k}+\sum\limits_{j=0}^{n-1}R_{j,k}
\biggr. \\
& \qquad \qquad \qquad \quad \biggl.+\sum_{\substack{j=1\\j\ne k}}^{N}\frac{2r_{n,j}(r_{n,k}-\lambda_k)}{t_j-t_k} \biggr]. \label{LHS}
\end{aligned}
\end{equation}

Inserting $A_n(z)$ and $B_n(z)$ given by \eqref{3.2} and \eqref{3.3} into  the right hand side of
(\ref{$S_{2}'$}), in view of \eqref{z(z+t)}, we obtain
\begin{equation}
	\begin{aligned}
		r.h.s.=&\beta_n \biggl[\frac{1}{z^2}\biggl(1-\sum_{k=1}^{N}R_{n,k} \biggr)\biggl(1-\sum_{k=1}^{N}R_{n-1,k} \biggr)+\sum_{k=1}^{N}\frac{R_{n,k}}{z+t_k} \sum_{j=1}^{N}\frac{R_{n-1,j}}{z+t_j} \biggr]\\
		&+\sum_{k=1}^{N}\frac{1}{z}\cdot \frac{\beta_n}{t_k}\biggl(R_{n,k}\biggl(1-\sum\limits_{k=1}^{N}R_{n-1,k}\biggr)+R_{n-1,k}\biggl(1-\sum\limits_{k=1}^{N}R_{n,k}\biggr) \biggr)\\
		&-\sum_{k=1}^{N}\frac{1}{z+t_k}\cdot \frac{\beta_n}{t_k}\biggl(R_{n,k}\biggl(1-\sum\limits_{k=1}^{N}R_{n-1,k}\biggr)+R_{n-1,k}\biggl(1-\sum\limits_{k=1}^{N}R_{n,k}\biggr) \biggr). \label{rhs1}	
	\end{aligned}
\end{equation}
Via an argument similar to the one used to derive \eqref{lhs1}, we simplify the second term in square bracket of the above equation to get
\begin{align*}
	\sum_{k=1}^{N}\frac{R_{n,k}}{z+t_k} \sum_{j=1}^{N}\frac{R_{n-1,j}}{z+t_j}=\sum\limits_{k=1}^{N}\frac{R_{n,k}R_{n-1,k}}{(z+t_k)^2}+\sum_{k\ne j}\frac{2R_{n,k}R_{n-1,j}}{t_j-t_k}\cdot \frac{1}{z+t_k}.
\end{align*}
Plugging it back into \eqref{rhs1} yields
\begin{equation}
	\begin{aligned}
r.h.s.=&\beta_n \biggl[\frac{1}{z^2}\biggl(1-\sum_{k=1}^{N}R_{n,k} \biggr)\biggl(1-\sum_{k=1}^{N}R_{n-1,k} \biggr)+\sum_{k=1}^{N}\frac{R_{n,k}R_{n-1,k}}{(z+t_k)^2} \biggr]\\
	&+\sum_{k=1}^{N}\frac{1}{z}\cdot \frac{\beta_n}{t_k}\biggl(R_{n,k}\biggl(1-\sum\limits_{k=1}^{N}R_{n-1,k}\biggr)+R_{n-1,k}\biggl(1-\sum\limits_{k=1}^{N}R_{n,k}\biggr) \biggr) \\
	&-\sum_{k=1}^{N}\frac{1}{z+t_k} \biggl[\frac{\beta_n}{t_k} \biggl(R_{n,k}\biggl(1-\sum\limits_{k=1}^{N}R_{n-1,k}\biggr)+R_{n-1,k}\biggl(1-\sum\limits_{k=1}^{N}R_{n,k}\biggr) \biggr)
	\biggr. \\
	& \qquad \qquad \qquad \quad  \biggl.
	-\sum_{\substack{j=1\\j\ne k}}^{N}\frac{2R_{n,k}R_{n-1,j}}{t_j-t_k} \biggr]. \label{rhs}	
	\end{aligned}
\end{equation}
Comparing the coefficients of $z^{-2}$ and $(z+t_k)^{-2}$ in \eqref{LHS} with the ones in \eqref{rhs}, we find the
following two equations
\begin{align}
	\biggl(n+\alpha+\sum_{k=1}^{N}r_{n,k}\biggr) \biggl(n+\sum_{k=1}^{N}r_{n,k}\biggr)&=\beta_n\biggl(1-\sum_{k=1}^{N}R_{n,k}\biggr) \biggl(1-\sum_{k=1}^{N}R_{n-1,k}\biggr),\label{S2'1}\\
	r_{n,k}^2-\lambda_kr_{n,k}&=\beta _nR_{n,k}R_{n-1,k}, \qquad k=1,\cdots,N.\label{S2'2}
\end{align}

Using \eqref{S1.1}-\eqref{S1.2} and \eqref{S2'1}-\eqref{S2'2}, we can express $\alpha_n$ and $\beta_n$ in terms of $\{R_{n,k},r_{n,k}\}$.

\begin{lemma} \label{alphabeta}
The recurrence coefficients are expressed in terms of the auxiliary quantities by
	\begin{align}
	 \alpha_n=2n+1+\alpha+\sum_{k=1}^{N}(\lambda_k-t_kR_{n,k}), \label{alpha_n}
    \end{align}
    \begin{align}
	\beta_n=\frac{1}{1-\sum\limits_{k=1}^{N}R_{n,k}} \left(n+\alpha+\sum\limits_{k=1}^{N}r_{n,k}\right) \left(n+\sum\limits_{k=1}^{N}r_{n,k}\right)+\sum\limits_{k=1}^{N}\frac{r_{n,k}^2-\lambda_k r_{n,k}}{R_{n,k}}.\label{beta_n}
	\end{align}	
\end{lemma}
\begin{proof}
By using \eqref{S1.2} to eliminate $r_{n+1,k}+r_{n,k}$ in \eqref{S1.1}, we obtain \eqref{alpha_n}. Rewriting \eqref{S2'1}, we have
	\begin{align*}
	\biggl(n+\alpha+\sum_{k=1}^{N}r_{n,k}\biggr) \biggl(n+\sum_{k=1}^{N}r_{n,k}\biggr)=\biggl(1-\sum_{k=1}^{N}R_{n,k}\biggr) \biggl(\beta_n-\sum_{k=1}^{N}\beta_n R_{n-1,k}\biggr).
	\end{align*}
According to \eqref{S2'2}, we replace the term $\beta_nR_{n-1,k}$ in the above equation by $\frac{r_{n,k}^2-\lambda_kr_{n,k}}{R_{n,k}}$, which leads us to \eqref{beta_n}.		
\end{proof}

\begin{remark}
Setting $\lambda_k=0$, $k=1,\cdots,N$ in \eqref{w(x)}, we find that our weight function is reduced to the classic Laguerre weight and the associated orthonormal polynomial $P_n(x)$ is monic Laguerre polynomial. In this case, by the definitions of $R_{n,k}$ and $r_{n,k}$ given by \eqref{R_n,k}-\eqref{r_n,k}, we know that $R_{n,k}=r_{n,k}=0$ and $\frac{r_{n,k}^2-\lambda_k r_{n,k}}{R_{n,k}}=0$. Consequently,
\eqref{alpha_n} and \eqref{beta_n} now read
\begin{align*}
	\alpha_n&=2n+1+\alpha,\\
	\beta_n&=n(n+\alpha),
\end{align*}
which are consistent with those given in \cite[(3.1.6)]{Ismail} for the orthonormal Laguerre polynomial and in \cite[section 5.1]{szego} for the classical Laguerre polynomial.
\end{remark}

Combining Lemma \ref{alphabeta} with \eqref{S1.2} and \eqref{S2'2}, we establish the following system of difference equations for the auxiliary quantities $\{R_{n,k}, r_{n,k}, k=1,\cdots,N\}$,
which can be iterated in $n$.

\begin{proposition}
	$\{R_{n,k}, r_{n,k}, k=1,\cdots,N\}$ satisfy the following system of difference equations
	\begin{align}
		 r_{n+1,k}=\lambda_k-(t_k+\alpha_n)R_{n,k}-r_{n,k},\qquad k=1,\cdots ,N,\label{de1}
	\end{align}
	\begin{align}
		 R_{n,1}=\frac{r_{n,1}(r_{n,1}-\lambda_1)}{\beta_nR_{n-1,1}}, \label{de2}
	\end{align}
	\begin{align}
		 R_{n,k}=\frac{r_{n,k}(r_{n,k}-\lambda_k)}{r_{n,1}(r_{n,1}-\lambda_k)} \cdot \frac{R_{n,1}R_{n-1,1}}{R_{n-1,k}}, \qquad k=2,\cdots ,N,\label{de3}
	\end{align}
	which can be iterated in $n$ with the initial conditions
	$$R_{0,k}=\frac{\lambda_k}{h_n}\int_{0}^{+\infty } \frac{x^\alpha{\rm e}^{-x} }{x+t_k}\cdot\prod_{k=1}^{N}(x+t_k)^{\lambda_k}dx, \qquad r_{0,k}=0,$$
	for $k=1,\cdots,N$. Here $\alpha_n$ and $\beta_n$ are given by \eqref{alpha_n} and \eqref{beta_n}.
\end{proposition}
\begin{proof}
Rewriting the equation \eqref{S1.2}, we get \eqref{de1}.
According to \eqref{S2'2} with $k=1$, i.e.
\begin{align}
	r_{n,1}^{2}-\lambda_1r_{n,1}=\beta _{n}R_{n,1}R_{n-1,1},\label{de4}
\end{align}
we get \eqref{de2}. Combining \eqref{de4} with \eqref{S2'2} for $k=2,\cdots ,N$ to eliminate $\beta_{n}$, we arrive at \eqref{de3}.
\end{proof}

To continue, according to \eqref{S2.1}-\eqref{S2.2}, we get the expression for $p(n,\vec{t}\,)$, the coefficient of $x^{n-1}$ in $P_n(x;\vec{t}\,)$,  in terms of the auxiliary quantities.
We will see in Section \ref{section5} that this formula plays an important role in the derivation of the second order PDE satisfied by the logarithmic derivative of the Hankel determinant.

\begin{lemma}
	$p(n,\vec{t}\,)$ is represented in terms of the auxiliary quantities by
	\begin{align}
		 p(n,\vec{t}\,)=-\beta_{n}-\sum_{k=1}^{N}t_kr_{n,k}\label{p(n,t)},
	\end{align}
where $\beta_n$ is given by \eqref{beta_n}.
\end{lemma}	
\begin{proof}
	Substituting \eqref{S2.2} into \eqref{S2.1} to eliminate $\beta _{n+1}R_{n+1,k}-\beta _{n}R_{n-1,k}$, $k=1,\cdots ,N$ gives us
	\begin{align}
		\alpha _n=\beta _{n+1}-\beta _n+\sum_{k=1}^{N}t_k(r_{n+1,k}-r_{n,k}).\nonumber
	\end{align}
	Replacing $n$ by $j$ in the above equation and summing it over $j$ from $0$ to $n-1$,  in view of the fact that $\beta _0=r_{0,k}=0$, we arrive at
	\begin{align}
		\sum_{j=0}^{n-1}\alpha _j=\beta _n+\sum_{k=1}^{N}t_kr_{n,k}.\label{alpha-beta}
	\end{align}
Combining it with \eqref{sum_alpha}, we get \eqref{p(n,t)}.
\end{proof}

\section{Toda Equations, Riccati Equations and Generalized  Painlev\'{e}  V Equation}	

By differentiating the orthogonality relation \eqref{or} with $m=n$ and $m=n-1$, we get differential relations, which enable us to establish Toda equations for the recurrence coefficients.
Combining them with the equations obtained in the previous section, we further deduce the Riccati equations satisfied by the auxiliary quantities.
\begin{lemma}
	The relationships between the derivatives of  $\ln{h_n(\vec{t}\,)}$, $p(n,\vec{t}\,)$ and the auxiliary quantities $\{R_{n,k}, r_{n,k} \}$ are given as follows
	\begin{align}
		\frac{\partial }{\partial {t_k}} \ln{h_n(\vec{t}\,)}&=R_{n,k}(\vec{t}\,), \label{dr1}\\
		\frac{\partial}{\partial {t_k}} p(n,\vec{t}\, )&=-r_{n,k}(\vec{t}\,), \label{dr2}
	\end{align}
	for $k=1,\cdots ,N.$ According to \eqref{beta} and \eqref{alpha}, we find
	\begin{align}
	\frac{\partial }{\partial {t_k}} \ln{\beta _n}&=R_{n,k}(\vec{t}\,)-R_{n-1,k}(\vec{t}\,), \label{dr3} \\
	\frac{\partial}{\partial {t_k}} \alpha _n&=r_{n+1,k}(\vec{t}\,)-r_{n,k}(\vec{t}\,), \label{dr4}
	\end{align}
	for $k=1,\cdots ,N.$
\end{lemma}	
\begin{proof}	
Making $m=n$ in \eqref{or}, we have
\begin{align*}
h_n(\vec{t}\, )=\int_{0}^{+\infty } P_{n}^{2}(x;\vec{t}\,)x^{\alpha}{\rm{e}}^{-x}\prod_{k=1}^{N}(x+t_k)^{\lambda_k}dx.
\end{align*}
Taking the derivative of $h_n(\vec{t}\,)$ with respect to $t_k$, we get
\begin{equation}
	\begin{aligned}
		\frac{\partial }{\partial {t_k}} h_n(\vec{t}\,)=&\int_{0}^{+\infty}2P_n(x;\vec{t}\,)\frac{\partial P_n(x;\vec{t}\,)}{\partial t_k}\cdot x^{\alpha}{\rm{e}}^{-x}\prod_{k=1}^{N}(x+t_k)^{\lambda_k}dx \\
		 &+\int_{0}^{+\infty}P_n^2(x;\vec{t}\,)x^{\alpha}{\rm{e}}^{-x}\cdot \frac{\partial}{\partial t_k}\prod_{j=1}^{N}(x+t_j)^{\lambda_j} dx. \label{h_n}	 
	\end{aligned}
\end{equation}
Applying $\frac{\partial }{\partial {t_k}}$ to $P_n(x;\vec{t}\,)$ yields
\begin{align*}
	\frac{\partial }{\partial {t_k}}P_n(x;\vec{t}\, )=\frac{\partial }{\partial {t_k}}(x ^n+p(n,\vec{t}\, )x^{n-1}+\cdots)=\frac{\partial }{\partial {t_k}}p(n,\vec{t}\, )\cdot P_{n-1}(x;\vec{t}\,)+\cdots.
\end{align*}
According to the orthogonality relation \eqref{or},  the first integral on the right hand side of \eqref{h_n}  is zero.
Noting that
\begin{align*}
\frac{\partial}{\partial t_k}\prod\limits_{j=1}^{N}(x+t_j)^{\lambda_j}=\frac{\lambda_k}{x+t_k}\prod\limits_{j=1}^{N}(x+t_j)^{\lambda_j},
\end{align*}
with the aid of \eqref{R_n,k}, hence \eqref{h_n} becomes
\begin{align}
\frac{\partial }{\partial {t_k}} h_n(\vec{t}\,)	 &=\lambda_k\int_{0}^{+\infty}\frac{P_{n}^{2}(x;\vec{t}\,)}{x+t_k}w(x;\vec{t}\,)dx \label{h_n1}  \\
&=h_n(\vec{t}\,)R_{n,k}(\vec{t}\,),\nonumber
\end{align}
which gives us \eqref{dr1}.

Making $m=n-1$ in \eqref{or}, in view of the orthogonality relation \eqref{or}, we get
\begin{align*}
0=\int_{0}^{+\infty } P_{n}(x;\vec{t}\,)P_{n-1}(x;\vec{t}\,)w(x;\vec{t}\,)dx.
\end{align*}
Differentiating it over $t_k$, via an argument similar to the one used to derive \eqref{h_n1}, we obtain
\begin{align*}
	0=&\int_{0}^{+\infty} \frac{\partial P_n(x;\vec{t}\,)}{\partial t_k}\cdot P_{n-1}(x;\vec{t}\,)w(x;\vec{t}\,)dx+\int_{0}^{+\infty}P_n(x;\vec{t}\,) \frac{\partial P_{n-1}(x;\vec{t}\,)}{\partial t_k}\cdot w(x;\vec{t}\,)dx \nonumber\\
	&+\int_{0}^{+\infty } P_{n}(x;\vec{t}\,)P_{n-1}(x;\vec{t}\,)\frac{\partial w(x;\vec{t}\,)}{\partial t_k}dx \nonumber\\
	=&h_{n-1}\frac{\partial }{\partial t_k}p(n,\vec{t}\,)+\lambda_k\int_{0}^{+\infty}\frac{P_{n}(x;\vec{t}\,)P_{n-1}(x;\vec{t}\,)}{x+t_k}w(x;\vec{t}\,)dx.
\end{align*}
According to the definition of $r_{n,k}$ given by \eqref{r_n,k}, we are led to \eqref{dr2}.
\end{proof}

\begin{proposition}
The recurrence coefficients satisfy the following Toda equations
\begin{align}
	&\delta \ln \beta_n=\alpha_{n-1}-\alpha_n+2,\label{te1}\\
	&(\delta-1) \alpha _n=\beta _{n}-\beta _{n+1}, \label{te2}
\end{align}
	where $\delta:=\sum\limits_{k=1}^{N}t_k \frac{\partial }{\partial {t_k}}.$
\end{proposition}	
\begin{proof}
Multiplying both sides of equation \eqref{dr3} by $t_k$ yields	
\begin{align*}
		t_k\frac{\partial }{\partial {t_k}} \ln{\beta _n}=t_k(R_{n,k}-R_{n-1,k}).
\end{align*}
Summing the above equation over $k$ from $1$ to $N$, we get
\begin{align*}
\delta \ln \beta_n=\sum_{k=1}^{N}t_kR_{n,k}-\sum_{k=1}^{N}t_kR_{n-1,k}.	 
\end{align*}
Eliminating the two summation terms by using \eqref{alpha_n} gives us \eqref{te1}.

Multiplying both sides of equation \eqref{dr4} by $t_k$ and summing over $k$ from $1$ to $N$, we have
\begin{align*}
\delta\alpha_n&=\sum_{k=1}^{N}t_kr_{n+1,k}-\sum_{k=1}^{N}t_kr_{n,k}\\
&=\alpha_{n}+\beta_{n}-\beta_{n+1},	
\end{align*}
where the second equality is due to \eqref{alpha-beta}. Hence we come to \eqref{te2}.
\end{proof}

\begin{theorem}
	The auxiliary quantities $\left\lbrace R_{n,k}, r_{n,k},k=1,\cdots ,N\right\rbrace $ satisfy the following Riccati equations
	\begin{align}
		\delta R_{n,k}=\biggl[ 2n+\alpha+t_k+\sum_{j=1}^{N}\biggl(\lambda_j-t_jR_{n,j}\biggr)\biggr ]R_{n,k}+2r_{n,k}-\lambda_k,\label{re1}
	\end{align}
	\begin{align}
			\delta r_{n,k}=\frac{r_{n,k}^2-\lambda_k r_{n,k}}{R_{n,k}}-\frac{R_{n,k}\biggl(n+\alpha+\sum\limits_{j=1}^{N}r_{n,j}\biggr) \biggl(n+\sum\limits_{j=1}^{N}r_{n,j}\biggr)}{1-\sum\limits _{j=1}^{N}R_{n,j}}-R_{n,k}\cdot\sum_{j=1}^{N}\frac{r_{n,j}^2-\lambda_jr_{n,j}}{R_{n,j}}, \label{re2}
	\end{align}
	for $k=1,\cdots,N$, where $\delta =\sum\limits_{j=1}^{N}t_j \frac{\partial }{\partial {t_j}}$.
\end{theorem}

\begin{proof}
First, we will make use of the following relations	
\begin{align}
	\frac{\partial}{\partial {t_k}} R_{n,j}
	&=\frac{\partial}{\partial {t_j}} R_{n,k}, \qquad k,j=1,\cdots,N,\label{re3}\\
	\frac{\partial}{\partial {t_k}} r_{n,j}
	&=\frac{\partial}{\partial {t_j}} r_{n,k}, \qquad k,j=1,\cdots,N,\label{re4}
\end{align}
where the first equation can be obtain by combining the fact that $\frac{\partial^2}{\partial t_jt_k}\ln{h_n} =\frac{\partial^2}{\partial t_kt_j}\ln{h_n}$ with \eqref{dr1}, and the second one is due to $\frac{\partial^2}{\partial t_jt_k}p(n,\vec{t}\, )=\frac{\partial^2}{\partial t_kt_j}p(n,\vec{t}\,)$ and \eqref{dr2}.

Inserting \eqref{alpha_n} into \eqref{dr4} leads us to
\begin{align*}
	r_{n+1,k}-r_{n,k}&=-\frac{\partial }{\partial t_k} \sum_{j=1}^{N}t_jR_{n,j}\\
	&=- R_{n,k}-\sum_{j=1}^{N}t_j\frac{\partial }{\partial t_k}R_{n,j}.
\end{align*}
Consequently, we have
\begin{align*}
	\sum_{j=1}^{N}t_j \frac{\partial }{\partial t_k}R_{n,j}=r_{n,k}-r_{n+1,k}-R_{n,k}.
\end{align*}
Using \eqref{S1.2} to eliminate $r_{n+1,k}$ in the above equation, in view of \eqref{re3}, we obtain
\begin{align*}
	\sum\limits_{j=1}^{N}t_j\frac{\partial }{\partial {t_j}}R_{n,k}=(t_k+\alpha_{n}-1)R_{n,k}+2r_{n,k}-\lambda_k.
\end{align*}
Getting rid of $\alpha_{n}$ in the above equation by using \eqref{alpha_n} yields \eqref{re1}.	
	
Substituting \eqref{p(n,t)} into \eqref{dr2}, we get
\begin{align*}
	r_{n,k}&=\frac{\partial }{\partial t_k}\beta_n+\frac{\partial }{\partial t_k} \sum_{j=1}^{N}t_jr_{n,j} \\
	&=\frac{\partial }{\partial t_k}\beta_n+r_{n,k}+\sum_{j=1}^{N}t_j\frac{\partial }{\partial t_k}r_{n,j},
\end{align*}
which, with the aid of \eqref{re4} gives us
\begin{align}
	\delta r_{n,k}&=-\frac{\partial }{\partial t_k}\beta_n\label{delta r} \\
	&=\beta_nR_{n-1,k}-\beta_nR_{n,k}\nonumber\\
	 &=\frac{r_{n,k}^2-\lambda_kr_{n,k}}{R_{n,k}}-\beta_nR_{n,k},\label{re5}
\end{align}
where the second equality is due to \eqref{dr3}, and the third one results from \eqref{S2'2}. Plugging \eqref{beta_n} into \eqref{re5}, we come to \eqref{re2}.		 
\end{proof}

At the end of this section, by solving $r_{n,k}$ from the first Riccati equation \eqref{re1} and plugging it into \eqref{re2}, in view of $\delta t_k=t_k$ and $\delta (t_kR_{n,k})=t_kR_{n,k}+t_k\cdot \delta R_{n,k}$, we construct the second order PDEs for $\{R_{n,k},k=1,\cdots,N\}$.

\begin{theorem}
The auxiliary quantities $\{R_{n,k},k=1,\cdots,N\}$ satisfy the following second order non-linear PDEs
\begin{equation}
	\begin{aligned}
\delta^2R_{n,k}=&t_kR_{n,k}-R_{n,k}\sum_{j=1}^{N}t_jR_{n,j}+\frac{(\delta R_{n,k})^2-\lambda_k^2}{2R_{n,k}}-R_{n,k}\sum_{j=1}^{N}\frac{(\delta R_{n,j})^2-\lambda_j^2}{2R_{n,j}}\\
&+\frac{1}{2}R_{n,k}\left(\biggl(t_k+\sum\limits_{j=1}^{N}\lambda_j\biggr)^2-\sum_{j=1}^{N}R_{n,j}\biggl(t_j+\sum_{k=1}^{N}\lambda_k\biggr)^2 \right)+R_{n,k}\sum_{j=1}^{N}\delta R_{n,j}\sum_{k=1}^{N}\lambda_k   \\
&+\biggl(2n+\alpha-\sum_{j=1}^{N}t_jR_{n,j}\biggr)t_kR_{n,k}	 -R_{n,k}\sum_{j=1}^{N}\biggl(2n+\alpha-\sum_{k=1}^{N}t_kR_{n,k} \biggr)t_jR_{n,j} \\
&+\frac{R_{n,k}}{2\biggl(\sum\limits_{j=1}^{N}R_{n,j}-1\biggr)}\left(\left(\sum\limits_{k=1}^{N}\delta R_{n,k}-\sum\limits_{j=1}^{N}\lambda_j\sum_{k=1}^{N}R_{n,k}+\sum_{k=1}^{N}\lambda_k\right)^2-\alpha^2\right), \label{PDEs}	
	\end{aligned}
\end{equation}
for $k=1,\cdots,N$, where $\delta =\sum\limits_{j=1}^{N}t_j \frac{\partial }{\partial {t_j}}$ and $\delta^2=\sum\limits_{j=1}^{N}t_j^2\frac{\partial^2}{\partial t_j^2}+2\sum\limits_{1\le k<j\le N}t_kt_j\frac{\partial^2}{\partial t_k\partial t_j} +\delta.$
\end{theorem}	
\begin{remark}
When $N=1$, the weight function $w(x;\vec{t}\,)$ given by \eqref{w(x)} now becomes $x^{\alpha}{\rm{e}}^{-x}(x+t_1)^{\lambda_1}$, and the corresponding Hankel determinant $D_n(t_1)$ was studied in \cite{ChenMckay}. The system \eqref{PDEs}  is simplified to
\begin{align*}
t_1^2R_{n,1}''+t_1R_{n,1}'=&t_1R_{n,1}-t_1R_{n,1}^2+\frac{1-R_{n,1}}{2R_{n,1}}\biggl( (t_1R_{n,1}')^2-\lambda_1^2 \biggr)+\frac{1}{2}(t_1+\lambda_1)^2R_{n,1}(1-R_{n,1}) \\
&+t_1(2n+\alpha-t_1R_{n,1})R_{n,1}(1-R_{n,1})+\lambda_1t_1R_{n,1}R_{n,1}'\\
&+\frac{R_{n,1}}{2(R_{n,1}-1)}\biggl((t_1R_{n,1}'-\lambda_1R_{n,1}+\lambda_1)^2-\alpha^2 \biggr).	
\end{align*}
Let $y_n:=\frac{R_{n,1}}{R_{n,1}-1}$, the above equation is transformed into
\begin{align*}
	y_n''=&\biggl(\frac{1}{2y_n}+\frac{1}{y_n-1} \biggr)(y_n')^2-\frac{y_n'}{t_1}+\frac{(y_n-1)^2}{t_1^2}\biggl(\frac{\alpha^2y_n}{2}-\frac{\lambda_1^2}{2y_n} \biggr)\\
	 &+(2n+\alpha+1+\lambda_1)\frac{y_n}{t_1}-\frac{1}{2}\frac{y_n(y_n+1)}{y_n-1},
\end{align*}
which is $P_V(\frac{\alpha^2}{2}, -\frac{\lambda_1^2}{2}, 2n+\alpha+1+\lambda_1, -\frac{1}{2})$ \cite[(C.41)]{Jimbo}. The second order ODEs for $R_{n,1}$ and $y_n$ that we presented above were not given in \cite{ChenMckay}.
\end{remark}

\section{Generalized $\sigma$-form of Painlev\'{e} V Equation} \label{section5}

In this section, we study the logarithmic derivative of the Hankel determinant, i.e.
\begin{align}
	\sigma_n(\vec{t}\,):=\delta \ln{D_n(\vec{t}\,)}, \label{sigma_n:}
\end{align}
where $\delta =\sum\limits_{j=1}^{N}t_j \frac{\partial }{\partial {t_j}}$.
We aim to establish the second order PDE for $\sigma_n(\vec{t}\,)$.

Inserting \eqref{D_n} i.e.
$D_n(\vec{t} \,)=\prod\limits_{k=0}^{n-1} h_k(\vec{t} \,)$ into \eqref{sigma_n:}, in view of \eqref{dr1} and \eqref{alpha_n}, we get
\begin{align*}
	 \sigma_n(\vec{t}\,)&=\sum_{k=0}^{n-1}\sum_{j=1}^{N}t_j\frac{\partial}{\partial t_j}\ln h_k\\
	&=\sum_{k=0}^{n-1}\sum_{j=1}^{N}t_jR_{k,j}\\
	&=-\sum_{k=0}^{n-1}\alpha_k+n\biggl(n+\alpha +\sum_{k=1}^{N}\lambda_k\biggr),
\end{align*}
which combined with the fact that $\sum\limits_{k=0}^{n-1}\alpha_k(\vec{t}\,)=-p(n,\vec{t}\,)$ (i.e. formula \eqref{sum_alpha}) gives us
\begin{align}
	\sigma_n(\vec{t}\,)=p(n,\vec{t}\,)+n\biggl(n+\alpha +\sum_{k=1}^{N}\lambda_k\biggr).\label{sigma_n}
\end{align}

We will use this relationship together with the conclusions established in the previous section to deduce the PDE satisfied by $\sigma_n(\vec{t}\,)$.
To achieve our goal, we first get the expression of $\sigma_n$ in terms of $\{R_{n,k}, r_{n,k}\}$.
Then, in turn, we express $\{R_{n,k}, r_{n,k}\}$ by $\sigma_n$ and its derivatives.

\begin{theorem}\label{theorem4.1}
$({\rm i})$
$\sigma_n(\vec{t}\,)$ is expressed in terms of the auxiliary quantities $\{R_{n,k}, r_{n,k}, k=1,\cdots,N\}$ by
	\begin{align}
		\sigma_n=n \biggl(n+\alpha +\sum_{k=1}^{N}\lambda_k \biggr)-\sum_{k=1}^{N}t_kr_{n,k}-\beta _n,\label{sigma-rR}
	\end{align}
where $\beta_n$ is given by \eqref{beta_n}.

$({\rm ii})$
The auxiliary quantities $\{R_{n,k}, r_{n,k}, k=1,\cdots,N\}$ are expressed in terms of $\sigma_n(\vec{t}\,)$ and its derivatives by
\begin{align}
	r_{n,k}=-\frac{\partial \sigma_n}{\partial{t_k}} ,\label{r-sigma}
\end{align}
\begin{align}
	R_{n,k}=\frac{-\sum\limits _{j=1}^{N} t_j\frac{\partial ^2\sigma_n}{\partial t_j\partial t_k}+sgn(R_{n,k}+R_{n-1,k})\sqrt{\Delta_k}}{2\biggl(-\sigma_n+\sum\limits_{k=1}^{N} t_k\frac{\partial \sigma_n }{\partial {t_k}} +n \biggl(n+\alpha +\sum\limits_{k=1}^{N}\lambda_k \biggr)\biggr)}, \label{R-sigma}
\end{align}
for $k=1,\cdots,N$, where sgn$(R_{n,k}+R_{n-1,k})$ is the sign function of $R_{n,k}+R_{n-1,k}$ which is $-1$ for $R_{n,k}+R_{n-1,k}<0$, $1$ for $R_{n,k}+R_{n-1,k}>0$ and $0$ for $R_{n,k}+R_{n-1,k}=0$.
Here $\Delta_k$ is defined by
\begin{align}
	\Delta_k=\biggl( \sum\limits _{j=1}^{N} t_j\frac{\partial ^2\sigma_n}{\partial t_j\partial t_k} \biggr)^2+4\biggl(-\sigma_n+\sum\limits_{k=1}^{N} t_k\frac{\partial \sigma_n }{\partial {t_k}} +n \biggl(n+\alpha +\sum\limits_{k=1}^{N}\lambda_k \biggr)\biggr)\frac{\partial\sigma_n}{\partial t_k} \biggl(\frac{\partial\sigma_n}{\partial t_k}+\lambda_k \biggr). \label{delta_k}
\end{align}	
\end{theorem}
\begin{proof}
Inserting the expression \eqref{p(n,t)} into \eqref{sigma_n} gives us \eqref{sigma-rR}.
Differentiating both side of  \eqref{sigma_n} with respect to $t_k$, in view of \eqref{dr2}, we arrive at \eqref{r-sigma}.
	
To achieve \eqref{R-sigma}, we rewrite \eqref{re5} as a quadratic algebraic equation in $R_{n,k}$
\begin{align}
	\beta_nR_{n,k}^2 +\delta r_{n,k}\cdot R_{n,k}-r_{n,k}(r_{n,k}-\lambda_k)=0.\label{qae}
\end{align}
Now we look at its discriminant
\begin{align}
	\Delta_k=(\delta r_{n,k})^2+4\beta _nr_{n,k}(r_{n,k}-\lambda_k). \label{delta}
\end{align}
Replacing $r_{n,k}(r_{n,k}-\lambda_k)$ in the above equation by $\beta_nR_{n,k}^2+\delta r_{n,k}\cdot R_{n,k}$, which is due to \eqref{re5}, we obtain
\begin{align*}
	\Delta_k=(\delta r_{n,k}+2\beta_{n}R_{n,k})^2
	\ge 0
\end{align*}
Hence, equation \eqref{qae} has two real roots given by
\begin{align}
	R_{n,k}=\frac{-\delta r_{n,k}\pm \sqrt{\Delta _k} }{2\beta _n}.\label{root}
\end{align}
Getting rid of the term $\delta r_{n,k}$ in the above formula by using \eqref{delta r}, with the aid of \eqref{dr3}, we come to
$$R_{n,k}+R_{n-1,k}=\pm \frac{ \sqrt{\Delta_k} }{2\beta _n}.$$
Hence the sign function before $\sqrt{\Delta_k}$ is determined by sgn$(R_{n,k}+R_{n-1,k})$.  Therefore \eqref{root} becomes
\begin{align}
	R_{n,k}=\frac{-\delta r_{n,k}+sgn(R_{n,k}+R_{n-1,k}) \sqrt{\Delta_k} }{2\beta _n}.\label{root1}
\end{align}
Substituting \eqref{r-sigma} into \eqref{sigma-rR},  we are able to express $\beta_n$ in terms of $\sigma_n$ and its derivative by
\begin{align}
	\beta_n=-\sigma_n+\sum\limits_{k=1}^{N} t_k\frac{\partial \sigma_n }{\partial {t_k}}+n \biggl(n+\alpha +\sum_{k=1}^{N}\lambda_k \biggr). \label{beta-sigma}
 \end{align}
Inserting it and \eqref{r-sigma} into \eqref{root1} and \eqref{delta}, we obtain \eqref{R-sigma} and \eqref{delta_k}.
\end{proof}

Now, we go ahead to derive the second order PDE for $\sigma_n$.
To achieve this, we first simplify \eqref{sigma-rR} and then
substitute \eqref{r-sigma}-\eqref{R-sigma} into the resulting equation.

\begin{theorem}
$\sigma_n(\vec{t}\,)$ satisfies the following second order PDE
\begin{equation}
	\begin{aligned}
	&\left[2\left(\sum\limits_{k=1}^{N} t_k\frac{\partial \sigma_n }{\partial {t_k}}-\sigma_n+n\biggl(n+\alpha+\sum\limits_{k=1}^{N}\lambda_k\biggr) \right)-\sum\limits_{k=1}^{N}sgn(R_{n,k}+R_{n-1,k})\sqrt{\Delta_k}  \right]^2 \\
	&-4\left(\sum\limits_{k=1}^{N} t_k\frac{\partial \sigma_n }{\partial {t_k}}-\sigma_n+n\biggl(n+\alpha+\sum\limits_{k=1}^{N}\lambda_k\biggr) \right)\left(n+\alpha-\sum\limits_{k=1}^{N}\frac{\partial\sigma_n}{\partial t_k} \right)\left(n-\sum\limits_{k=1}^{N}\frac{\partial\sigma_n}{\partial t_k} \right)\\
	&-\left(\sum\limits_{k=1}^{N}\sum\limits _{j=1}^{N} t_j\frac{\partial ^2\sigma_n}{\partial t_k\partial t_j} \right)^2=0, \label{sigma-pde}
	\end{aligned}
\end{equation}
where $\Delta_k$ is given by \eqref{delta_k}.
\end{theorem}
\begin{proof}
Plugging \eqref{beta_n} into \eqref{sigma-rR}, we get
\begin{align}
\sigma_n=n \biggl(n+\alpha +\sum_{k=1}^{N}\lambda_k \biggr)-\sum_{k=1}^{N}t_kr_{n,k}-\frac{\left(n+\alpha+\sum\limits_{k=1}^{N}r_{n,k}\right) \left(n+\sum\limits_{k=1}^{N}r_{n,k}\right)}{1-\sum\limits_{k=1}^{N}R_{n,k}} -\sum\limits_{k=1}^{N}\frac{r_{n,k}^2-\lambda_k r_{n,k}}{R_{n,k}}. \label{sigma-1}	
\end{align}
We first look at the last term on the right hand side of the above equation. According to \eqref{re5}, we have
\begin{align}
	\frac{r_{n,k}^2-\lambda_k r_{n,k}}{R_{n,k}}&=\delta r_{n,k}+\beta_nR_{n,k} \nonumber\\
	&=\frac{1}{2}\delta r_{n,k}+\frac{1}{2}sgn(R_{n,k}+R_{n-1,k}) \sqrt{\Delta_k}, \label{last term}
\end{align}
where the last equality is due to \eqref{root1}. Substituting \eqref{last term} back into \eqref{sigma-1}, and replacing $R_{n,k}$ by using \eqref{root1}, after simplification,
we get
\begin{align*}
	&\biggl(2\beta_n+\sum\limits_{k=1}^{N}\biggl(\delta r_{n,k}-sgn(R_{n,k}+R_{n-1,k})\sqrt{\Delta_k} \biggr)\biggr)\sigma_n\\
	=&\biggl(2\beta_n+\sum\limits_{k=1}^{N}\biggl(\delta r_{n,k}-sgn(R_{n,k}+R_{n-1,k})\sqrt{\Delta_k} \biggr) \biggr) \biggl(n\biggl(n+\alpha+\sum\limits_{k=1}^{N}\lambda_k\biggr)-\sum_{k=1}^{N}t_kr_{n,k} \biggr)\\
	&-\left(\beta_n+\frac{1}{2} \sum\limits_{k=1}^{N}\biggl(\delta r_{n,k}-sgn(R_{n,k}+R_{n-1,k})\sqrt{\Delta_k} \biggr) \right)\sum\limits_{k=1}^{N}\biggl(\delta r_{n,k}+sgn(R_{n,k}+R_{n-1,k}) \sqrt{\Delta_k}\biggr)\\
	 &-2\beta_n\left(n+\alpha+\sum\limits_{k=1}^{N}r_{n,k}\right) \left(n+\sum\limits_{k=1}^{N}r_{n,k}\right).
\end{align*}
Eliminating $\sigma_n$ in the above equation by using \eqref{sigma-rR}, after simplification, we are led to
\begin{align*}
	 \biggl(2\beta_n-\sum\limits_{k=1}^{N}sgn(R_{n,k}+R_{n-1,k})\sqrt{\Delta_k} \biggr)^2-\biggl(\sum\limits_{k=1}^{N}\delta r_{n,k} \biggr)^2&\\
	 -4\beta_n\biggl(n+\alpha+\sum\limits_{k=1}^{N}r_{n,k} \biggr)\biggl(n+\sum\limits_{k=1}^{N}r_{n,k} \biggr)&=0.
\end{align*}
Substituting \eqref{r-sigma} and \eqref{beta-sigma} into the above equation, we arrive at \eqref{sigma-pde}.
\end{proof}

\begin{remark}
When $N=1$, with $t_1$ replaced by $t$ and $\lambda_1$ by $\lambda$, equation \eqref{sigma-pde} is reduced to
\begin{align*}
	 (t\sigma_n'-\sigma_n+n\lambda+(2n+\alpha+\lambda)\sigma_n')^2=(t\sigma_n'')^2+4(t\sigma_n'-\sigma_n+n(n+\alpha+\lambda))((\sigma_n')^2+\lambda \sigma_n'),
\end{align*}
	which agrees with $(69)$ of \cite{ChenMckay}.
Let $H_n(t):=\sigma_n(t)-n\lambda$, then $H_n$ satisfies the $\sigma$-form of a Painlev\'{e} V equation  of \cite[(C.45)]{Jimbo} with $\nu_0=0$, $\nu_1=\lambda$, $\nu_2=-n$ and $\nu_3=-n-\alpha$, i.e.
\begin{align*}
	 (tH_n'')^2=(H_n-tH_n'+2(H_n')^2+(\lambda-2n-\alpha)H_n' )^2-4H_n'(\lambda+H_n')(-n+H_n')(-n-\alpha+H_n').
\end{align*}	
\end{remark}

\begin{remark}
When $N=2$, the PDE \eqref{sigma-pde} is reduced to
\begin{align*}
	\sigma_n=&-2\biggl(\frac{\partial\sigma_n}{\partial t_1}\biggr) \biggl(\frac{\partial\sigma_n}{\partial t_2} \biggr)+(2n+\alpha+t_1+\lambda_1)\frac{\partial\sigma_n}{\partial t_1}+ (2n+\alpha+t_2+\lambda_2)\frac{\partial\sigma_n}{\partial t_2} \\
	 &\pm\biggl(\sqrt{\Delta_1}-\sqrt{\Delta_2}\biggr)-\frac{\left( \biggl( t_1\frac{\partial^2}{\partial t_1^2}+t_2\frac{\partial^2}{\partial t_1\partial t_2}\biggr)\sigma_n \right)\left( \biggl( t_2\frac{\partial^2}{\partial t_2^2}+t_1\frac{\partial^2}{\partial t_1\partial t_2}\biggr)\sigma_n \right)+\sqrt{\Delta_1\Delta_2} }{2\biggl(t_1\frac{\partial\sigma_n}{\partial t_1}+t_2\frac{\partial\sigma_n}{\partial t_2}-\sigma_n+n(n+\alpha)  \biggr)},
\end{align*}	
which coincides with $(46)$	of \cite{ChenHaq}, where $\lambda_2=-\lambda_1$ and $t$, $T$, $N_s$ were used in place of $t_1$, $t_2$, $\lambda_1$.
\end{remark}

\section{Double Scaling Analysis at the Hard Edge}

\subsection{The Generalized $\sigma$-form of a Painlev\'{e} III Equation}
In this section, we study the behavior of the Hankel determinant $D_n(\vec{t} \,)$ at the hard edge.
It is well known that the Laguerre polynomials $P_n^{(\alpha)}(x)$ has the following property (see \cite[Theorem 8.1.1]{szego} )
\begin{align*}
\lim_{n \to \infty} n^{-\alpha}P_n^{(\alpha)}\biggl(1-\frac{z^2}{2n^2}\biggr)=\biggl( \frac{z}{2} \biggr)^{-\alpha}J_{\alpha}(z),
\end{align*}
where $J_{\alpha}(\cdot)$ is the Bessel function of the first kind of order $\alpha$. It motivates us to consider the double scaling that
$n\to \infty$ and $t_k\to 0^+$
such that $s_k=4nt_k$ for $k=1,\cdots,N$ is fixed.

Assuming
$\vec{t}=\frac{\vec{s}}{4n}$, $\vec{s}=(s_1,s_2,\dots,s_N)$,
we have
\begin{align*}
	\delta =\sum\limits_{k=1}^{N}t_k \frac{\partial }{\partial {t_k}}=\sum\limits_{k=1}^{N}s_k \frac{\partial }{\partial {s_k}}.
\end{align*}
Define
\begin{align}
	\sigma(\vec{s}\,):=\lim_{n \to \infty}\sigma_n\biggl(\frac{\vec{s}}{4n}\biggr).\label{sigma-s}
\end{align}
By using the finite $n$ results, i.e. Theorem \ref{theorem4.1} and the Riccati equation \eqref{re1},
we build direct relationships between the scaled $\{R_{n,k}, r_{n,k}, k=1,\cdots,N\}$ and $\sigma(\vec{s}\,)$.
\begin{proposition}
The auxiliary quantities $\{R_{n,k}, r_{n,k}, k=1,\cdots,N\}$ are scaled as below
\begin{align}
	&\lim_{n \to \infty}\frac{r_{n,k}\biggl(\frac{\vec{s}}{4n}\biggr)}{n}=-4\frac{\partial \sigma(\vec{s}\,)}{\partial {s_k}}, \label{limr}\\
	&\lim_{n \to \infty}R_{n,k}\biggl(\frac{\vec{s}}{4n}\biggr)=4\frac{\partial \sigma(\vec{s}\,)}{\partial {s_k}}=:R_k(\vec{s}\,). \label{limR}
\end{align}
\end{proposition}
\begin{proof}
We use the relationships between $\{R_{n,k}, r_{n,k}\}$ and $\sigma_n$ in Theorem \ref{theorem4.1} to get the desired results. Substituting $\vec{t}=\frac{\vec{s}}{4n}$ and $t_k=\frac{s_k}{4n}$ into \eqref{r-sigma}, we get
\begin{align}
	r_{n,k}\biggl(\frac{\vec{s}}{4n}\biggr)=-4n\cdot  \frac{\partial }{\partial {s_k}}\sigma_n\biggl(\frac{\vec{s}}{4n}\biggr).\label{r-s}
\end{align}
Dividing its both sides by $n$ and taking the limit $n\to \infty$, in view of \eqref{sigma-s}, we arrive at \eqref{limr}.

Inserting $\vec{t}=\frac{\vec{s}}{4n}$ and $t_k=\frac{s_k}{4n}$ into \eqref{R-sigma}  yields
\begin{align}
R_{n,k}\biggl(\frac{\vec{s}}{4n}\biggr)=\frac{4n\sum\limits _{j=1}^{N} s_j\frac{\partial ^2}{\partial s_j\partial s_k}\sigma_n\biggl(\frac{\vec{s}}{4n}\biggr) \pm \sqrt{\Delta_k}}{2\left(-\sigma_n\biggl(\frac{\vec{s}}{4n}\biggr)+\sum\limits_{k=1}^{N} s_k\frac{\partial}{\partial {s_k}}\sigma_n\biggl(\frac{\vec{s}}{4n}\biggr) +n \biggl(n+\alpha +\sum\limits_{k=1}^{N}\lambda_k \biggr)\right)}, \label{R-s}
\end{align}
for $k=1,\cdots,N$, where $\Delta_k$ is defined by \eqref{delta_k} and now reads
\begin{align*}
\Delta_k=\left( 4n\sum\limits _{j=1}^{N} s_j\frac{\partial ^2\sigma_n}{\partial s_j\partial s_k} \right)^2
+16n\left(-\sigma_n+\sum\limits_{k=1}^{N} s_k\frac{\partial\sigma_n}{\partial {s_k}} +n \biggl(n+\alpha +\sum\limits_{k=1}^{N}\lambda_k \biggr)\right)\frac{\partial\sigma_n}{\partial s_k} \biggl(4n\frac{\partial\sigma_n}{\partial s_k}+\lambda_k \biggr).
\end{align*}	
Taking the limit $n\to \infty$ in \eqref{R-s}, we get
\begin{align*}
\lim_{n \to \infty}R_{n,k}\biggl(\frac{\vec{s}}{4n}\biggr)=\pm4\sqrt{\biggl(\frac{\partial\sigma}{\partial s_k} \biggr)^2},
\end{align*}
which indicates that $R_{n,k}\biggl(\frac{\vec{s}}{4n}\biggr)=O\biggl(\frac{\partial\sigma}{\partial s_k}\biggr)$ as $n\to\infty$. Therefore, replacing $r_{n,k}$ in \eqref{re1} by using \eqref{r-s} and $t_k$ by $\frac{s_k}{4n}$, and taking the limit $n\to \infty$ , we come to \eqref{limR}.
\end{proof}

Replacing $t_k$ by $\frac{s_k}{4n}$ in the PDEs for $R_{n,k}(\vec{t}\,)$, i.e. \eqref{PDEs} and sending $n$ to $\infty$,
in view of \eqref{limR}, we obtain the second order PDEs satisfied
by $R_k(\vec{s}\,)$.
\begin{theorem}
$\{R_k(\vec{s}\,),k=1,\cdots,N\}$ satisfy the following system of  PDEs
\begin{equation}
	\begin{aligned}
\delta^2R_k=&\frac{(\delta R_k)^2-\lambda_k^2}{2R_k}-R_k\sum_{j=1}^{N}\frac{(\delta R_j)^2-\lambda_j^2}{2R_j}+\frac{1}{2}\biggl(\sum_{j=1}^{N}\lambda_j \biggr)^2\biggl(1-\sum_{j=1}^{N}R_j\biggr)R_k\\
&+\frac{R_k}{2\biggl(\sum\limits_{j=1}^{N}R_j-1 \biggr)}\left(\biggl(\sum\limits_{k=1}^{N}\delta R_k-\sum\limits_{j=1}^{N}\lambda_j\sum\limits_{k=1}^{N}R_k+\sum\limits_{k=1}^{N}\lambda_k \biggr)^2-\alpha^2 \right)\\
&+\sum_{k=1}^{N}\lambda_k\cdot R_k \sum_{j=1}^{N}\delta R_j +\frac{R_k}{2}\biggl(s_k-\sum_{j=1}^{N}s_jR_j \biggr), \label{R_k}
	\end{aligned}
\end{equation}
for $k=1,\cdots,N$, where $\delta =\sum\limits_{j=1}^{N}s_j \frac{\partial }{\partial {s_j}}$ and $\delta^2=\sum\limits_{j=1}^{N}s_j^2\frac{\partial^2}{\partial s_j^2}+2\sum\limits_{1\le k<j\le N}s_ks_j\frac{\partial^2}{\partial s_k\partial s_j} +\delta.$
\end{theorem}
\begin{remark}
When $N=1$, with $s_1$ and $\lambda_1$ replaced by $s$ and $\lambda$ respectively,  the system \eqref{R_k} is
reduced to
\begin{align*}
2(s^2R_1''+sR_1')=&\frac{(sR_1')^2-\lambda^2}{R_1}-(sR_1')^2+(\lambda)^2+(s+\lambda^2)R_1(1-R_1)+2\lambda sR_1R_1'\\
&+\frac{R_1}{R_1-1}\biggl((sR_1'-\lambda R_1+\lambda)^2-\alpha^2\biggr).
\end{align*}
Denoting $Y_1(s):=\frac{R_1(s)}{R_1(s)-1}$, then $Y_1$ satisfies the following equation
\begin{align*}
	Y_1''=\biggl(\frac{1}{2Y_1}+\frac{1}{Y_1-1} \biggr)(Y_1')^2-\frac{Y_1'}{s}+\frac{(Y_1-1)^2}{s^2}\biggl(\frac{\alpha^2}{2}Y_1-\frac{\lambda^2}{2Y_1} \biggr)+\frac{Y_1}{2s},
\end{align*}
which is $P_V\biggl(\frac{\alpha^2}{2}, -\frac{\lambda^2}{2}, -\frac{1}{2}, 0\biggr)$ \cite{Jimbo}.
In this sense, we may treat equations \eqref{R_k} as the generalization of the Painlev\'{e} V equation.
\end{remark}

Substituting $\vec{t}=\frac{\vec{s}}{4n}$ and $t_k=\frac{s_k}{4n}$ into \eqref{sigma-pde}, and taking the limit $n\to \infty$, in view of \eqref{sigma-s}, we establish the second order PDE satisfied by
$\sigma(\vec{s}\,)$.
\begin{theorem}
$\sigma(\vec{s}\,)$ satisfies the following PDE
\begin{equation}
	\begin{aligned}
	&\frac{16\biggl(\sum\limits_{k=1}^{N}s_k \frac{\partial ^2\sigma}{\partial s_k^2}+\sum\limits_{1\le j< k\le N}(s_j+s_k)\frac{\partial ^2\sigma}{\partial s_j\partial s_k}  \biggr)^2-\alpha^2}{4\sum\limits_{k=1}^{N}\frac{\partial \sigma}{\partial s_k}-1}-\sum_{k=1}^{N}\frac{16\biggl(\sum\limits_{j=1}^{N}s_j \frac{\partial ^2\sigma}{\partial s_j\partial s_k} \biggr)^2-\lambda_k^2}{4\frac{\partial \sigma}{\partial s_k}}\\
	&+4\delta \sigma-4\sigma-\biggl(\alpha+\sum\limits_{j=1}^{N}\lambda_j \biggr)^2=0. \label{sigma}	
	\end{aligned}
\end{equation}
\end{theorem}
\begin{remark}
When $N=1$, by using $s$ and $\lambda$ in place of $s_1$ and $\lambda_1$, \eqref{sigma} becomes
\begin{align*}
\frac{16(s\sigma'')^2-\alpha^2}{4\sigma'-1}
-\frac{4(s\sigma'')^2-\frac{\lambda^2}{4}}{\sigma'}+4s\sigma'-4\sigma-(\alpha+\lambda)^2=0.
\end{align*}
By introducing $s=4t$,
it can be transformed into the following Jimbo-Miwa-Okamoto $\sigma$-form of the Painlev\'{e} III equation (see \cite[equation (3.13)]{jimbo} ) satisfied by $\sigma(4t)$:
\begin{align*}
	 (t\sigma'')^2=4\sigma'(\sigma-t\sigma')(\sigma'-1)+((\alpha+\lambda)\sigma'-\lambda)^2,
\end{align*}
where the derivative is with respect to $t$. In this sense, we may regard \eqref{sigma} as the generalized  $\sigma$-form of a Painlev\'{e} III equation.
\end{remark}

\subsection{Equilibrium Density}

When the dimension of the Hermitian matrices from the deformed unitary ensemble associated with the weight function $w(x)$ tends to $\infty$,
the eigenvalues can be approximated as a continuous fluid with an equilibrium density $\psi(x)$.
If the potential function ${\rm v}(x)=-\ln{w(x)}$ is convex, then $\psi(x)$ was shown to be supported on a single interval $(a,b)$ \cite{ChenIsmail2}.
According to Dyson's Coulomb fluid theory \cite{Dyson} and the results in \cite{ChenIsmail2},
the equilibrium density is obtained by minimizing the free energy functional
\begin{align*}
	F[\psi]:=\int_{a}^{b} \psi(x){\rm v}(x)dx-\int_{a}^{b}\int_{a}^{b}\psi(x)\ln{|x-y|}\psi(y)dxdy
\end{align*}
subject to the normalization condition
\begin{align}
	\int_{a}^{b}\psi(x)dx=n. \label{sigma(x)}	
\end{align}
Upon minimization, the density $\psi(x)$ satisfies the following integral equation
\begin{align}
	{\rm v}(x)-2\int_{a}^{b}\ln{|x-y|}\psi(y)dy=A, \quad x\in (a,b), \label{A}
\end{align}
where $A$ is the Lagrange multiplier. Taking the derivative of \eqref{A} with respect to $x$ gives us the following singular integral equation
\begin{align*}
	{\rm v}'(x)-2P\int_{a}^{b}\frac{\psi(y)}{x-y}dy=0, \quad x\in (a,b),
\end{align*}
where $P$ denotes the Cauchy principal value. The solution of  this equation, subject to the boundary condition $\psi(a)=\psi(b)=0$, is given by
\begin{align}
	\psi(x)=\frac{\sqrt{(b-x)(x-a)}}{2\pi^2} P\int_{a}^{b}\frac{{\rm v}'(x)-{\rm v}'(y)}{(x-y)\sqrt{(b-y)(y-a)}}dy, \label{sigma-solution}
\end{align}
with a supplementary condition
\begin{align}
	\int_{a}^{b}\frac{{\rm v}'(x)dx}{\sqrt{(b-x)(x-a)}}=0. \label{condition1}
\end{align}
Substituting \eqref{sigma-solution} into \eqref{sigma(x)} leads us to
\begin{align}
	\int_{a}^{b}\frac{x{\rm v}'(x)dx}{\sqrt{(b-x)(x-a)}}=2n\pi.
	\label{condition2}
\end{align}
For our problem, ${\rm v}(x)=x-\alpha\ln{x}-\sum\limits_{k=1}^{N}\lambda_k \ln{(x+t_k)}$.
Hence we obtain
\begin{align*}
	{\rm v}''(x)=\frac{\alpha}{x^2}+\sum\limits_{k=1}^{N}\frac{\lambda_k}{(x+t_k)^2}.
\end{align*}
To ensure ${\rm v}(x)$ to be convex, i.e. ${\rm v}''(x)>0$, we assume in this section $\lambda_k\ge0$ for $k=1,\cdots,N$.

Now we proceed to calculate $\psi(x)$ and $A$ by using
\eqref{sigma-solution} and \eqref{A} respectively. With the help of the integral identities, we get the following expressions.
\begin{proposition}
For $\lambda_k\ge0, k=1,\cdots,N$, the equilibrium density for the eigenvalues of the deformed Laguerre unitary ensemble associated with the weight function \eqref{w(x)} is given by
\begin{align}
	 \psi(x)=\frac{\sqrt{(b-x)(x-a)}}{2\pi}\left(\frac{\alpha}{x\sqrt{ab}}+\sum_{k=1}^{N}\frac{\lambda_k}{x+t_k}\cdot \frac{1}{\sqrt{(a+t_k)(b+t_k)}} \right),\label{sigmax}
\end{align}
and the  Lagrange multiplier in the integral equation \eqref{A} satisfied by $\psi(x)$ reads
\begin{equation}
	\begin{aligned}
	 A=&\frac{a+b}{2}-\alpha\ln{\left(\frac{a+b+2\sqrt{ab}}{4} \right)}-2n\ln{\biggl(\frac{b-a}{4}\biggr)}\\
	 &-\sum\limits_{k=1}^{N}\lambda_k\ln{\left(\frac{a+b+2t_k+2\sqrt{(a+t_k)(b+t_k)}}{4}\right)}. \label{a}	
	\end{aligned}
\end{equation}
\end{proposition}
\begin{proof}
Substituting \eqref{v'} into \eqref{sigma-solution}, with the aid of the following integral identities for $0<a<b$
(see \cite[equation (263)]{ChenMckay} )
\begin{align*}
	 &\int_{a}^{b}\frac{dx}{x\sqrt{(b-x)(x-a)}}=\frac{\pi}{\sqrt{ab}}, \\
	 &\int_{a}^{b}\frac{dx}{(x+t)\sqrt{(b-x)(x-a)}}=\frac{\pi}{\sqrt{(a+t)(b+t)}},
\end{align*}
we get \eqref{sigmax}.

Next, we calculate the Lagrange multiplier $A$.  Multiplying both sides of \eqref{A} by $\frac{1}{\sqrt{(b-x)(x-a)}}$
and integrating it with respect to $x$ from $a$ to $b$, in view of \eqref{v(x)}, we have
\begin{align}
	 \int_{a}^{b}\frac{x-\alpha\ln{x}-\sum\limits_{k=1}^{N}\lambda_k\ln{(x+t_k)}}{\sqrt{(b-x)(x-a)}}dx-2\int_{a}^{b}\psi(y)dy\int_{a}^{b}\frac{\ln{|x-y|}}{\sqrt{(b-x)(x-a)}}dx=A\pi. \label{Api}	
\end{align}
We look at the first integral in the above equality. Based on the following integral identities for $0<a<b$
(see equations (265) and (248) in \cite{ChenMckay} )
\begin{align*}
	 &\int_{a}^{b}\frac{xdx}{\sqrt{(b-x)(x-a)}}=\frac{a+b}{2}\pi, \\
	&\int_{a}^{b}\frac{\ln{x} dx}{\sqrt{(b-x)(x-a)}}=2\pi \ln{\left(\frac{\sqrt a+\sqrt b}{2}\right)}, \\
	 &\int_{a}^{b}\frac{\ln{(x+t)}dx}{\sqrt{(b-x)(x-a)}}=2\pi \ln{\left(\frac{\sqrt{a+t}+\sqrt{b+t}}{2}\right)},\;  t>0,
\end{align*}
we arrive at
\begin{align}
	 &\int_{a}^{b}\frac{x-\alpha\ln{x}-\sum\limits_{k=1}^{N}\lambda_k\ln{(x+t_k)}}{\sqrt{(b-x)(x-a)}}dx \nonumber\\
	=&\frac{a+b}{2}\pi-2\pi \alpha \ln{\left(\frac{\sqrt a+\sqrt b}{2}\right)}-2\pi\sum\limits_{k=1}^{N}\lambda_k\ln{\left(\frac{\sqrt{a+t_k}+\sqrt{b+t_k}}{2}\right)}. \label{integral1}
\end{align}
Now we focus on the integral for $x$ in the second term on the left hand side of \eqref{Api}.
According to the following formula for principal value integral
\begin{align*}
	P\int_{a}^{b}\frac{dx}{(x-y)\sqrt{(b-x)(x-a)}}=0,
\end{align*}
we arrive at
\begin{align}
	 \int_{a}^{b}\frac{\ln{|x-y|}}{\sqrt{(b-x)(x-a)}}dx=C, \label{C}
\end{align}
where $C$ is a constant independent of $y$. Hence, we replace $y$ by $b$ in \eqref{C} and get
\begin{align}
	\int_{a}^{b}\frac{\ln{|x-y|}}{\sqrt{(b-x)(x-a)}}dx
	=&\int_{a}^{b}\frac{\ln{(b-x)}}{\sqrt{(b-x)(x-a)}}dx \quad\quad   \nonumber\\
	=&\int_{0}^{1}\frac{\ln{(b-a)}}{\sqrt{(1-S)S} }dS+\int_{0}^{1}\frac{\ln{(1-S)}}{\sqrt{(1-S)S} }dS \nonumber
	\\
	=&\pi\ln{(b-a)}-2\pi\ln{2}, \label{integral2}
\end{align}
where $S=\frac{x-a}{b-a}$, and to get the last equality we make use of the following identities for $0<a<b$
(see equations (264) and (248) in \cite{ChenMckay} ):
\begin{align*}
	&\int_{0}^{1}\frac{dS}{\sqrt{(1-S)S} }=\pi, \\
	&\int_{0}^{1}\frac{\ln{(1-S)}}{\sqrt{(1-S)S} }dS=-2\pi\ln{2}.
\end{align*}
Plugging \eqref{integral1} and \eqref{integral2} back into \eqref{Api}, on account of \eqref{sigma(x)}, we are led to \eqref{a}.
\end{proof}

\begin{remark}
Inserting ${\rm v}(x)$ given by \eqref{v(x)} into \eqref{condition1} and \eqref{condition2},  we get two equations for $a$ and $b$
\begin{equation}
	\begin{aligned}
&\frac{\alpha}{\sqrt{ab}}+\sum\limits_{k=1}^{N}\frac{\lambda_k}{\sqrt{(a+t_k)(b+t_k)}}-1=0,\\
&2n+\alpha+\sum\limits_{k=1}^{N}\lambda_k-\sum\limits_{k=1}^{N}\frac{\lambda_k t_k}{\sqrt{(a+t_k)(b+t_k)}}-\frac{a+b}{2}=0.	 \label{N}
	\end{aligned}
\end{equation}
When $N=1$, these two identities are reduced to
\begin{align}
	 &\frac{\alpha}{\sqrt{ab}}+\frac{\lambda_1}{\sqrt{(a+t_1)(b+t_1)}}-1=0, \label{1}\\
	&2n+\alpha+\lambda_1-\frac{\lambda_1 t_1}{\sqrt{(a+t_1)(b+t_1)}}-\frac{a+b}{2}=0.\nonumber
\end{align}
By using them, it was shown in \cite{ChenFilipuk} that the corresponding monic orthogonal polynomials satisfy the confluent Heun equation. We summarize the derivation as follows. By eliminating the term $\frac{\lambda_1}{\sqrt{(a+t_1)(b+t_1)}}$ from the above two equations, a direct relationship between $\sqrt{ab}$ and $\frac{a+b}{2}$ was found. In view of it, and by clearing the square root in \eqref{1}, an algebraic equation was derived for $\frac{a+b}{2}$, from which the asymptotic expansion of the recurrence coefficient $\beta_n$ followed. And consequently, the confluent Heun equation for the monic orthogonal polynomials was deduced.

For our problem with $N$ general,  it is hard to eliminate the terms
$((a+t_k)(b+t_k))^{-\frac{1}{2}}$ for $k=1,\cdots,N$ from \eqref{N} to build a relationship between $\sqrt{ab}$ and $\frac{a+b}{2}$. Therefore, the subsequent arguments for $N=1$ to derive the confluent Heun equation are not applicable to our problem.
\end{remark}
\begin{remark}
	When $\lambda_k=0$ for $k=1,\cdots,N$, \eqref{sigmax} becomes
	\begin{align}
		\psi(x)=\frac{\alpha}{2\pi\sqrt{ab}} \cdot \frac{\sqrt{(b-x)(x-a)}}{x},\quad  x\in(a,b). \label{psi}	
	\end{align}
And it follows from \eqref{N} that
\begin{align}
	ab=\alpha^2, \qquad a+b=2\alpha+4n. \label{ab}
\end{align}
 Substituting them into \eqref{psi} and letting $x=4ny$ in the resulting expression, by sending $n$ to $\infty$, we get
\begin{align*}
	\psi(x)=\frac{1}{2\pi}\sqrt{\frac{1-y}{y}}, \qquad  y\in(0,1),	
\end{align*}
which  coincides with (19.1.11) of \cite{Mehta}. Here $y\in(0,1)$ is obtained as below.

Write $\hat{a}:=\frac{a}{4n}$ and $\hat{b}:=\frac{b}{4n}$. Since $x=4ny\in(a,b)$, we have $y\in(\hat{a},\hat{b} )$. Replacing $a$ by $4n\hat{a}$ and $b$ by $4n\hat{b}$ in \eqref{ab}, after simplification, we get
\begin{align*}
	\hat{a}\cdot \hat{b}=\frac{\alpha^2}{16n^2}, \qquad 	 \hat{a}+\hat{b}=\frac{\alpha}{2n}+1,
\end{align*}
which, as $n\to\infty$, gives us $\hat{a}\cdot \hat{b}=0$ and $\hat{a}+\hat{b}=1$. Consequently, we find that $\hat{a}=0$ and $\hat{b}=1$, i.e. $y\in(0,1)$.
\end{remark}

\section*{Acknowledgments}
This work was supported by National Natural Science Foundation of China under grant numbers 12101343 and 12371257, and by Shandong Provincial Natural Science Foundation with project number ZR2021QA061.


\begin{thebibliography}{99}

\bibitem{Charlier2019}
C. Charlier, Asymptotics of Hankel determinants with a one-cut regular potential and Fisher-Hartwig singularities,
Int. Math. Res. Notices {\bf2019} (2019), 7515-7576.
	
\bibitem{CharlierDeano}
C. Charlier and A. Dea\~{n}o, Asymptotics for Hankel determinants associated to a Hermite weight with a varying discontinuity, Symmetry Integr. Geom. {\bf14} (2018), 018 (43 pages).


\bibitem{Charlier2021}
C. Charlier and R. Gharakhloo, Asymptotics of Hankel determinants with a
Laguerre-type or Jacobi-type potential and Fisher-Hartwig singularities,
Adv. Math. {\bf383} (2021), 107672 (69 pages).


\bibitem{H.Chen}
H. Chen, M. Chen, G. Blower, and Y. Chen, Single-user MIMO system, Painlev\'{e} transcendents, and double scaling, J. Math. Phys. {\bf58} (2017), 123502 (23 pages).





\bibitem{ChenFeigin}
Y. Chen and M. Feigin,
Painlev\'{e} IV and degenerate Gaussian unitary ensembles, J. Phys. A:Math. Gen. {\bf39} (2006), 12381-12393.



\bibitem{ChenFilipuk}
Y. Chen, G. Filipuk and L. Zhan, Orthogonal polynomials, asymptotics, and Heun equations, J. Math. Phys. {\bf60} (2019), 113501 (34 pages).

\bibitem{ChenHaq}
Y. Chen, N. Haq and M. McKay, Random matrix models, double-time Painlev\'{e} equations, and wireless relaying,  J. Math. Phys. {\bf54} (2013), 063506 (55 pages).



\bibitem{ChenIsmail}
Y. Chen and M. Ismail, Ladder operators and differential equations for orthogonal polynomials, J. Phys. A: Math. Gen. {\bf30} (1997), 7817-7829.


\bibitem{ChenIsmail2}
Y. Chen and M. Ismail, Thermodynamic relations of the Hermitian matrix ensembles, J. Phys. A: Math. Gen. {\bf30} (1997), 6633-6654.




\bibitem{ChenLyu}
Y. Chen and S. Lyu, Gaussian unitary ensembles with jump discontinuities, PDEs and the coupled Painlev\'{e} IV system,  Recent Progress in Special Functions (Editor
 G. Filipuk), accepted by Contemporary Mathematics, Amer. Math. Soc., to appear
 in Nov. 2024 (arXiv: 2304.04127).

\bibitem{ChenMckay}
Y. Chen and M. McKay, Coulumb fluid, Painlev\'{e} transcendents, and the information theory of MIMO system, IEEE Trans. Inf. Theory
{\bf58} (2012), 4594-4634.


\bibitem{Chihara}
T. Chihara, An introduction to orthogonal polynomials, Dover, New York, 1978.

\bibitem{DaiXuZhang}
D. Dai, S. Xu and L. Zhang, Gaussian unitary ensembles with pole singularities near the soft edge and a system of coupled Painlev\'{e} XXXIV equations,
Ann. Henri Poincare {\bf20} (2019), 3313-3364.


\bibitem{Dyson}
F. J. Dyson, Statistical theory of the energy levels of complex systems, I, II, III, J.Math. Phys. {\bf3} (1962) 140-156, 157-165, 166-175.

\bibitem{DeiftIts}
P. Deift and A. Its, Asymptotics of Toeplitz, Hankel, and Toeplitz+Hankel determinants with Fisher-Hartwig singularities, Ann. Math. {\bf174} (2011), 1243-1299.


\bibitem{DingMin}
Y. Ding and C. Min,
Semi-classical orthogonal polynomials associated with a modified Gaussian weight, Results Math. {\bf73} (2024), 102 (17 pages).


\bibitem{Ismail}
M. Ismail, Classical and Quantum Orthogonal Polynomials in One Variable, Encyclopedia of Mathematics and its Applications 98, Cambridge University Press, Cambridge, 2005.


\bibitem{Jimbo}
M. Jimbo and T. Miwa,
Monodromy perserving deformation of linear ordinary differential equations with rational coefficients. II, Physica D
{\bf2} (1981), 407-448.

\bibitem{jimbo}
M. Jimbo, Monodromy problem and the boundary condition for some Painlev\'{e} equations,
Publ. RIMS, Kyoto Univ. {\bf18} (1982), 1137-1161.



\bibitem{Kundu1}
N. Kundu, S. Dash, M. McKay and R. Mallik, Channel Estimation and Secret Key Rate Analysis of MIMO Terahertz Quantum Key Distribution, IEEE Trans. Commun. {\bf70} (2022), 3350-3363.



\bibitem{Liu}
F. Liu, C. Masouros, A. Petropulu, H. Griffiths and L. Hanzo, Joint radar and communication design: applications, state-of-the-art, and the road ahead, IEEE Trans. Commun. {\bf68} (2020), 3834-3862.


\bibitem{Lopez-Martinez}
F. Lopez-Martinez, E. Martos-Naya, J. Paris and A. Goldsmith,
Eigenvalue dynamics of a central Wishart matrix with application to MIMO systems, IEEE Trans. Inf. Theory {\bf61} (2015), 2693-2707.

\bibitem{LyuChenXu}
S. Lyu, Y. Chen and S. Xu, Laguerre unitary ensembles with jump discontinuities, PDEs and the coupled Painlev\'{e} V system, Physica D {\bf449} (2023), 133755 (14 pages).


\bibitem{LyuGriffinChen}
S. Lyu, J. Griffin and Y. Chen, The Hankel determinant associated with a singularly perturbed Laguerre unitary ensemble, J. Nonlinear Math. Phys. {\bf26} (2019), 24-53.

\bibitem{Magnus}
A. P. Magnus, Painlev\'{e}-type differential equations for the recurrence coefficients of semi-classical orthogonal polynomials,
J. Comput. Appl. Math. {\bf57} (1995), 215-237.


\bibitem{Mehta}
M. Mehta, Random Matrices, third edition, Elsevier, New York, 2004.





\bibitem{MinChenNPB1}
C. Min, S. Lyu and Y. Chen,  Painlev\'{e} III' and the Hankel determinant generated by
a singularly perturbed Gaussian weight, Nucl. Phys. B {\bf936} (2018), 169-188.




\bibitem{MinChenNPB2}
C. Min and Y. Chen,  Painlev\'{e} V and the Hankel determinant for a singularly
perturbed Jacobi weight, Nucl. Phys. B {\bf961} (2020), 115221 (25 pages).

\bibitem{MinChenStud}
C. Min and Y. Chen, Differential, difference, and asymptotic relations for Pollaczek-Jacobi type orthogonal
 polynomials and their Hankel determinants,
Stud. Appl. Math. {\bf147} (2021), 390-416.




\bibitem{MinChenJMP}
C. Min and Y. Chen, Hankel determinant and orthogonal polynomials for a perturbed Gaussian weight: From finite $n$ to large $n$ asymptotics, J. Math. Phys. {\bf64} (2023), 083503 (20 pages).

\bibitem{MinWang}
C. Min and L. Wang, Asymptotics of the smallest eigenvalue distributions of Freud unitary ensembles,  arXiv: 2402.15190.



\bibitem{MuLyu}
X. Mu and S. Lyu, Hankel determinants for a Gaussian weight with Fisher-Hartwig singularities and generalized Painlev\'{e} IV equation, J. Phys. A: Math. Theor. {\bf56} (2023), 475201 (25 pages).


\bibitem{Narmanlioglu}
O. Narmanlioglu and M. Uysal, Multi-user massive MIMO visible light communications with limited pilot transmission, IEEE Trans. Wirel. Commun. {\bf21} (2022), 4197-4211.


\bibitem{QiQian}
Y. Qi, R. Qian and N. Li, Coulomb gas analogy: a statistical physics approach to performance analysis of MIMO systems, IEEE Trans. Veh. Technol. {\bf68} (2019), 1984-1988.



\bibitem{szego}
G. Szeg{\"o},  Orthogonal Polynomials, American Mathematical Society Colloquium Publications, Vol. 23, New York, 1939.


\bibitem{Telatar}
E. Telatar, Capacity of multi-antenna Gaussian channels, Eur. Trans. Telecomm. {\bf10} (1999), 585-596.


\bibitem{TulinoVerdu}
A. Tulino and S. Verd\'{u}, Random matrix theory and wireless communications, Found. Trends Commun. Inf. Theory  {\bf1} (2004), 1-182.



\bibitem{vanassche}
W. Van Assche, Orthogonal Polynomials and Painlev\'{e} Equations, Cambridge University Press, New York, 2018.




\bibitem{WuXu}
X. Wu and S. Xu, Gaussian unitary ensemble with jump discontinuities and the coupled Painlev\'{e} II and IV systems, Nonlinearity {\bf34} (2021), 2070-2115.

\bibitem{XuZhao}
S. Xu and Y. Zhao, Gap probability of the circular unitary ensemble with a Fisher-Hartwig singularity and the coupled Painlev\'{e} V system,
Commun. Math. Phys. {\bf377} (2020), 1545-1596.


\bibitem{ZhuChen}
J. Yu, S. Chen, C. Li, M. Zhu, and Y. Chen, Painlev\'{e} V and confluent Heun equations associated with a perturbed Gaussian
unitary ensemble, J. Math. Phys. {\bf64} (2023), 083501 (22 pages).


\end{thebibliography}
\end{document}